\documentclass[11pt,letterpaper]{article} %

\usepackage{etoolbox}
\newbool{icalp}
\boolfalse{icalp} %

\usepackage{amsmath}
\usepackage{amssymb}
\usepackage{amsthm}
\usepackage{bbm}
\usepackage{thm-restate}

\ifbool{icalp}{  
    \usepackage{hyperref}
}{
    \usepackage[margin=1in]{geometry}
    \usepackage[utf8]{inputenc}
    \newtheorem{theorem}{Theorem}
    \newtheorem{lemma}[theorem]{Lemma}
    
    \newtheorem{proposition}[theorem]{Proposition}
    \newtheorem{definition}[theorem]{Definition}
    
    \newtheorem{corollary}[theorem]{Corollary}

    \newtheorem{claim}[theorem]{Claim}
    
    \newenvironment{claimproof}[1][\proofname]{
        \pushQED{\qed}%
        \normalfont %
        \trivlist
        \item\relax{\itshape#1{.}}\hspace\labelsep\ignorespaces
    }{%
        
        \popQED\endtrivlist%
        
    }

    \usepackage[pdfpagelabels,bookmarks=false]{hyperref}
    \hypersetup{colorlinks, linkcolor=darkblue, citecolor=darkgreen, urlcolor=darkblue}
    \usepackage[capitalize, nameinlink]{cleveref}
    \usepackage{times}
    \usepackage[table]{xcolor}
    \usepackage{comment}
    \excludecomment{CCSXML}
    
}

\usepackage{url}
\usepackage{array}

\usepackage{setspace}

\usepackage{lmodern}

\usepackage{upgreek}

\usepackage{complexity}

\usepackage{moreenum}
\usepackage{mathtools}
\usepackage[super]{nth}

\usepackage{bm}

\usepackage{xspace}

\setlist[enumerate]{nosep,topsep=0.1em}
\setlist[enumerate,1]{label=(\roman*), leftmargin=2.2em}
\setlist[itemize]{nosep,topsep=0.3em}

\usepackage[textsize=footnotesize, color=blue!30!white]{todonotes}

\usepackage{xfrac}

\usepackage{tikz}
\usetikzlibrary{calc}
\usetikzlibrary{shapes.geometric}
\usetikzlibrary{arrows}
\usetikzlibrary{decorations.pathreplacing, decorations.pathmorphing}

\usepackage{tcolorbox}
\tcbuselibrary{skins}

\usepackage{float}
\usepackage{graphicx}
\graphicspath{{graphics/}}
\makeatletter
\newcommand{\appendtographicspath}[1]{%
  \g@addto@macro\Ginput@path{#1}%
}
\makeatother

\usepackage[vlined,ruled,algo2e,linesnumbered]{algorithm2e}
\SetAlgoSkip{bigskip}
\SetAlgoInsideSkip{smallskip}
\SetAlCapHSkip{0.2em}
\SetCustomAlgoRuledWidth{0.95\textwidth}

\usepackage{mdframed}

\crefname{theorem}{Theorem}{Theorems}
\Crefname{lemma}{Lemma}{Lemmas}
\Crefname{claim}{Claim}{Claims}
\Crefname{fact}{Fact}{Facts}
\Crefname{remark}{Remark}{Remarks}
\Crefname{observation}{Observation}{Observations}
\Crefname{line}{Line}{Lines}
\crefalias{AlgoLine}{line}
\Crefname{algocf}{Algorithm}{Algorithms}

\definecolor{darkblue}{rgb}{0,0,0.38}
\definecolor{darkred}{rgb}{0.8,0,0}
\definecolor{darkgreen}{rgb}{0.1,0.35,0}

\newcommand{\OPT}{\ensuremath{\mathrm{OPT}}}

\renewcommand{\epsilon}{\varepsilon}

\def\Ascr{\mathcal{A}}

\def\Cscr{\mathcal{C}}

\def\Iscr{\mathcal{I}}
\def\Jscr{\mathcal{J}}
\def\Lscr{\mathcal{L}}

\def\Wscr{\mathcal{W}}

 \newcommand\apxfac{\ensuremath{1.29}\xspace}

\ifbool{icalp}{  

\hideLIPIcs
\bibliographystyle{plainurl}%

\title{Better-Than-\texorpdfstring{$\frac{4}{3}$}{4/3}-Approximations for Leaf-to-Leaf Tree and Connectivity Augmentation}

\titlerunning{Leaf-to-Leaf Tree and Connectivity Augmentation} 

\author{Anonymous}{ }{}{}{}

\authorrunning{Anonymous} %

\Copyright{Anonymous} %



\category{} 

\relatedversion{} %
\EventEditors{}
\EventNoEds{2}
\EventLongTitle{49th International Colloquium on Automata, Languages and Programming (ICALP 2022)}
\EventShortTitle{ICALP 2022}
\EventAcronym{ICALP}
\EventYear{2022}
\EventDate{July 4--8, 2022}
\EventLocation{Paris, France}
\EventLogo{}
\SeriesVolume{}
\ArticleNo{}
}{

\title{Better-Than-\texorpdfstring{$\frac{4}{3}$}{4/3}-Approximations for Leaf-to-Leaf Tree and Connectivity Augmentation%
\thanks{This project received funding from Swiss National Science Foundation grant 200021\_184622 and the European Research Council (ERC) under the European Union's Horizon 2020 research and innovation programme (grant agreement No 817750).}
}

\author{
Federica Cecchetto\thanks{
Department of Mathematics, ETH Zurich, Zurich, Switzerland.
Email: \href{mailto:federica.cecchetto@ifor.math.ethz.ch}%
{federica.cecchetto@ifor.math.ethz.ch}.
}
\and
Vera Traub\thanks{
Department of Mathematics, ETH Zurich, Zurich, Switzerland.
Email: \href{mailto:vera.traub@ifor.math.ethz.ch}%
{vera.traub@ifor.math.ethz.ch}.
}
\and
Rico Zenklusen\thanks{
Department of Mathematics, ETH Zurich, Zurich, Switzerland.
Email: \href{mailto:ricoz@ethz.ch}%
{ricoz@ethz.ch}.}
}

\date{}

\bibliographystyle{alpha}
}

\begin{document}

\maketitle

\ifbool{icalp}{  
}{
\thispagestyle{empty}
\addtocounter{page}{-1}

\begin{tikzpicture}[overlay, remember picture, shift = {(current page.south east)}]
\begin{scope}[shift={(-1.5,1.8)}]
\def\hd{1.5}
\node at (-2*\hd,0) {\includegraphics[height=0.35cm]{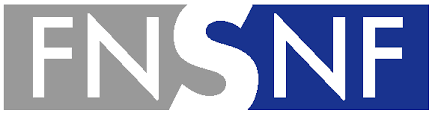}};
\node at (-\hd,0) {\includegraphics[height=0.7cm]{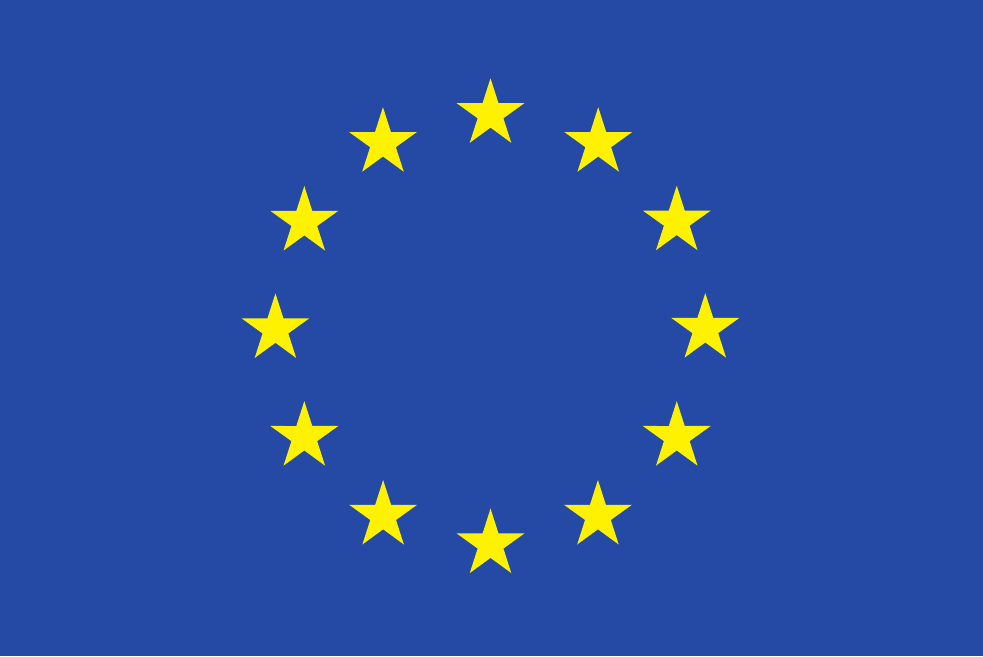}};
\node at (-0.2*\hd,0) {\includegraphics[height=0.8cm]{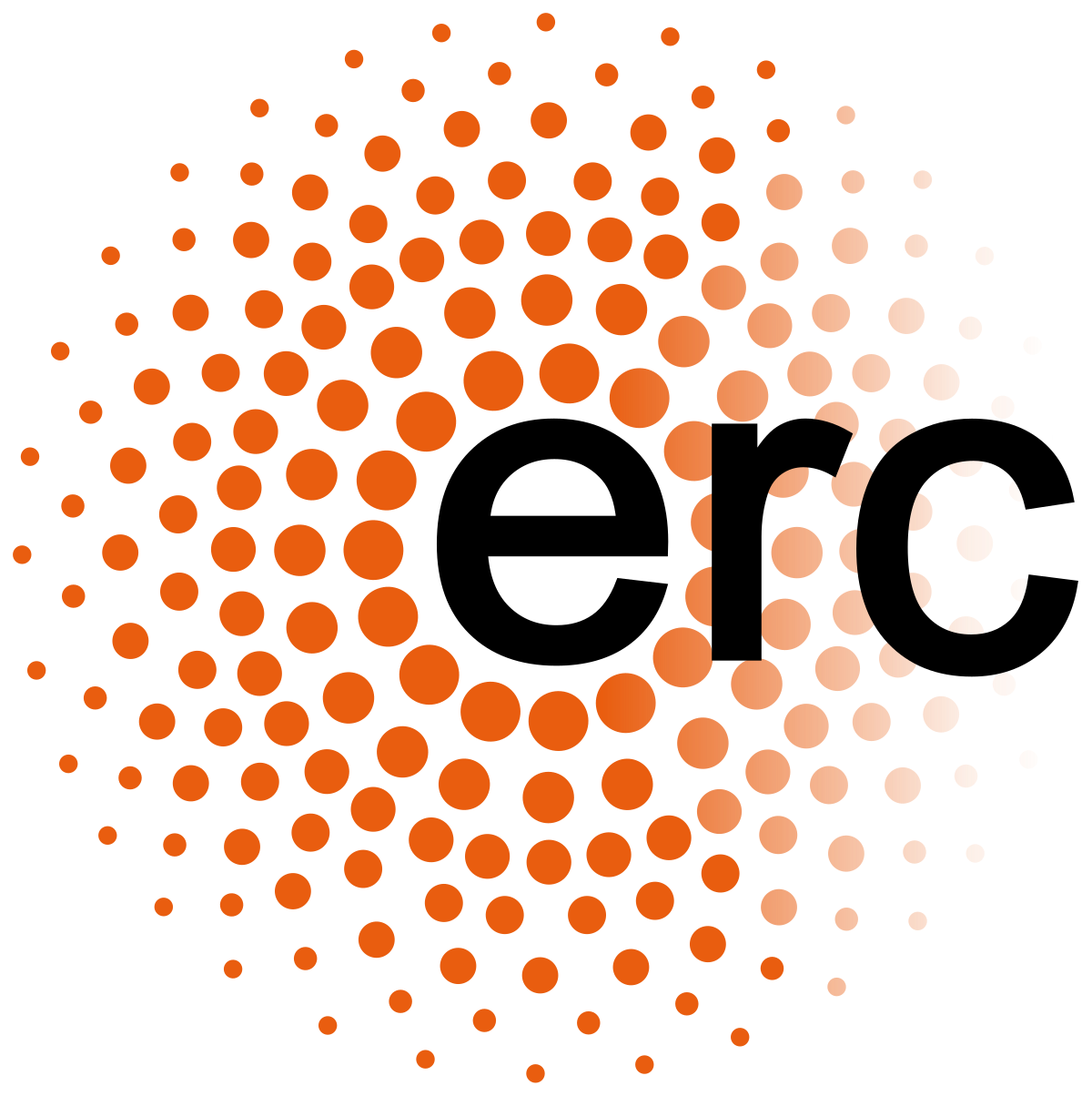}};
\end{scope}
\end{tikzpicture}
}

\begin{abstract}
The Connectivity Augmentation Problem (CAP) together with a well-known special case thereof known as the Tree Augmentation Problem (TAP) are among the most basic Network Design problems.
There has been a surge of interest recently to find approximation algorithms with guarantees below $2$ for both TAP and CAP, culminating in the currently best approximation factor for both problems of $1.393$ through quite sophisticated techniques.

We present a new and arguably simple matching-based method for the well-known special case of leaf-to-leaf instances.
Combining our work with prior techniques, we readily obtain  a $(\sfrac{4}{3}+\epsilon)$-approximation for Leaf-to-Leaf CAP by returning the better of our solution and one of an existing method.
Prior to our work, a $\sfrac{4}{3}$-guarantee was only known for Leaf-to-Leaf TAP instances on trees of height $2$.
Moreover, when combining our technique with a recently introduced stack analysis approach, which is part of the above-mentioned $1.393$-approximation, we can further improve the approximation factor to $1.29$, obtaining for the first time a factor below $\frac{4}{3}$ for a nontrivial class of TAP/CAP instances.
\end{abstract}

\section{Introduction}\label{sec:intro}

The Connectivity Augmentation Problem (CAP) is one of the most elementary Network Design problems.
It asks to increase the edge-connectivity of a graph by one unit in the most economical way by adding edges/links from a given set.
Formally, one is given a graph $G=(V,E)$ and an additional link set $L\subseteq \begin{psmallmatrix} V\\ 2 \end{psmallmatrix}$, and the task is to determine a smallest size set of links $U\subseteq L$ such that the edge-connectivity of $(V,E\cup U)$ is strictly larger than that of $G$.
A famous special case of CAP is the Tree Augmentation Problem (TAP), where the given graph $G$ is a spanning tree.
Already TAP is well-known to be $\APX$-hard (see~\cite{kortsarz_2004_hardness}, which presents an extension of a construction used in~\cite{frederickson_1981_approximation} to prove $\NP$-hardness), even on trees of diameter $5$ with all links going between pairs of leaves, i.e., \emph{leaf-to-leaf instances}.
This motivated the search for strong constant-factor approximations.
There has been extensive work during the last decades on the approximability of both TAP and CAP, and also their weighted counterparts where each link has a cost and instead of minimizing the size of $U$ the goal is to minimize its cost.

For CAP (and therefore also TAP), multiple $2$-approximations have been known for a long time, including through classical techniques like primal-dual algorithms and iterative rounding (see~\cite{frederickson_1981_approximation,khuller_1993_approximation,goemans_1994_improved,jain_2001_factor}).
It was an important stepping stone to reach algorithms with an approximation guarantee below $2$.
During the last decades, a long line of research led to multiple approaches that beat the approximation factor $2$ for TAP and CAP through the introduction of a rich set of techniques%
~\cite{%
adjiashvili_2018_beating,%
cheriyan_2018_approximating_a,%
cheriyan_2018_approximating_b,%
cheriyan_2008_integrality,%
cohen_2013_approximation,%
even_2009_approximation,%
fiorini_2018_approximating,%
frederickson_1981_approximation,%
grandoni_2018_improved,%
khuller_1993_approximation,%
kortsarz_2016_simplified,%
kortsarz_2018_lp-relaxations,%
nagamochi_2003_approximation,%
nutov_2020,%
grandoni_2018_improved,%
cecchetto_2021_bridging%
}.
This led to the current state-of-the art approximation factor of $1.393$~\cite{cecchetto_2021_bridging}, which is currently the best one for both TAP and CAP.
The factor of $1.393$ was obtained through a combination of quite sophisticated techniques, and the obtained factor is neither natural nor is it likely to be the ``right'' answer, i.e., it seems very likely that approximation algorithms with better approximation guarantees should exist.

A natural question we are interested in, is whether there may be a clean algorithm leading to a $\sfrac{4}{3}$-approximation.
This factor appears in other, related Network Design problems, in particular in the (unweighted) $2$-Edge-Connected Spanning Subgraph problem ($2$-ECCS).
In $2$-ECCS, one starts with an empty graph $G$ and the task is to pick a smallest number of links to obtain a $2$-edge-connected graph spanning all vertices.
Hence, instead of starting from a spanning tree to obtain a $2$-edge-connected graph as in TAP, one starts with an empty graph.
Recent advances on $2$-ECCS led to the best-known approximation factor of $\sfrac{4}{3}$~\cite{sebo_2014_shorter,hunkenschroder_2019_approximation}.
In the context of TAP and CAP, we still lack appropriate techniques to reach such factors, and progress along this line has only been achieved for quite restricted special cases.
More precisely, for TAP, a factor of $\sfrac{4}{3}$ is known to be achievable if we are given an optimal solution to the natural LP relaxation, known as the \emph{cut-LP}, that has the additional property of being half-integral~\cite{cheriyan_1992_2-coverings} or, more generally, fulfills that each non-zero entry is at least $\sfrac{1}{2}$~\cite{iglesias_2018_coloring}. 
However, the \emph{cut-LP} is in general not a half-integral LP~\cite{cheriyan_1992_2-coverings} and it may not contain any optimal point where each non-zero is at least $\sfrac{1}{2}$.
Moreover, an approach for TAP was presented in~\cite{maduel_2010_covering} that leads to a $\sfrac{4}{3}$-approximation for Leaf-to-Leaf TAP instances on trees of height at most $2$.
Finally, for CAP, we are unaware of any nontrivial class of instances where such factors, or even factors below the currently best $1.393$-approximation, are known.
Note that, even for TAP, natural special cases for which approximation factors below $\sfrac{4}{3}$ can be achieved are unknown to the best of our knowledge.

The goal of this paper is to make first progress in this regard for the case of Leaf-to-Leaf CAP (and therefore also TAP) instances.
(We formally define Leaf-to-Leaf CAP instances in Section~\ref{sec:preliminaries} and show their relation to Leaf-to-Leaf TAP and also Leaf-to-Leaf Cactus Augmentation, which is a well-known connection on which we heavily rely later on.)
We think of Leaf-to-Leaf CAP instances as an appealing class because of the following reasons.
First, they comprise a large family of nontrivial TAP/CAP instances, which have already been studied both in the context of TAP~\cite{maduel_2010_covering} and CAP~\cite{nutov_2021_approximation}.
Second, the best known hardness results for TAP/CAP are based on leaf-to-leaf instances.

We also note that, for the weighted settings of TAP/CAP, any instance can be reduced in an approximation-preserving way to a weighted leaf-to-leaf one; hence, the difference between leaf-to-leaf and general instances vanishes in the weighted settings.%
\footnote{Consider a general weighted TAP instance $G=(V,E)$ with links $L\subseteq \begin{psmallmatrix} V\\ 2\end{psmallmatrix}$ and weights $w:L\to \mathbb{R}_{\geq 0}$.
To transform it into a leaf-to-leaf instance, one can do the following operation for each vertex $v$:
(i) Add two new vertices $w^1_v, w^2_v$ to the graph and connect each of them only to $v$; hence, $w^1_v$ and $w^2_v$ are leaves;
(ii) For each link $\ell\in L$ that has $v$ as an endpoint, replace the endpoint $v$ by $w^1_v$;
(iii) Add a new link $\{w^1_v,w^2_v\}$ of cost $0$.
One can easily observe that the original instance is equivalent to the obtained leaf-to-leaf one, because, after including the newly added $0$-cost links, one falls back to the original instance.
An analogous construction works to reduce weighted CAP to weighted Leaf-to-Leaf CAP.
}

\subsection{Preliminaries}\label{sec:preliminaries}

Consider a CAP instance $(G=(V,E),L)$ on a graph $G$ that is $k$-edge-connected (but not $(k+1)$-edge-connected).
Hence, the goal is to add a smallest number of links $U\subseteq L$ such that $(V,E\cup U)$ is $(k+1)$-edge-connected.
By Menger's Theorem, a set $U\subseteq L$ leads to a $(k+1)$-connected graph $(V,E\cup U)$ if and only if each minimum cut in $G$ is crossed by at least one link of $U$.
$(G,L)$ is a \emph{Leaf-to-Leaf CAP} instance if, for each link $\{u,v\}\in L$, both $u$ and $v$ are contained in a minimal minimum cut. 
The \emph{minimal minimum cuts} in $G$ are minimum cuts $A\subseteq V$ for which no other minimum cut $B \subseteq V$ satisfies $B\subseteq A$.
Classic uncrossing results imply that the minimal minimum cuts of a graph form a family of disjoint sets.
Note that in case of a tree, these cuts are singleton cuts only containing a single leaf vertex.
Thus, a leaf-to-leaf instance for a tree indeed maps to only having links with both endpoints being leaves.

For CAP, it is often significantly more convenient to first use a well-known reduction to the \emph{Cactus Augmentation Problem} (CacAP), which is the special case of CAP where the underlying graph is a cactus, i.e., it is a connected graph with each edge being contained in a unique cycle.
See Figure~\ref{fig:LLCacAP_example} for an example.
\begin{figure}[h!]
\begin{center}
\begin{tikzpicture}[%
scale=0.55,
lks/.style={line width=1.3pt, blue, densely dashed},
ts/.style={every node/.append style={font=\scriptsize}},
ls/.style={every node/.append style={rectangle, fill=black!10}},
]

\begin{scope}[every node/.style={thick,draw=black,fill=white,circle,minimum size=6, inner sep=2pt}]

\begin{scope}[ls]
\node  (4) at (16.70,-7.96) {};
\node  (6) at (19.98,-8.18) {};
\node  (8) at (23.24,-5.32) {};
\node  (9) at (21.72,-4.38) {};
\node (10) at (22.50,-7.60) {};
\node (11) at (24.70,-7.04) {};
\node (14) at (10.76,-6.44) {};
\node (15) at (13.28,-7.18) {};
\node (16) at (14.80,-7.72) {};
\node (18) at (17.64,-2.90) {};
\node (19) at (23.26,-4.04) {};
\node (20) at (24.28,-3.00) {};
\node (22) at (13.66,-2.54) {};
\node (23) at (11.30,-3.68) {};
\node (24) at (13.52,-4.28) {};
\end{scope}

\node  (1) at (17.06,-4.94) {};
\node  (2) at (17.46,-6.32) {};
\node  (3) at (19.62,-4.74) {};
\node  (5) at (18.56,-7.16) {};
\node  (7) at (20.92,-6.34) {};
\node (12) at (12.36,-5.38) {};
\node (13) at (15.52,-5.84) {};
\node (17) at (20.14,-3.28) {};
\node (21) at (15.78,-3.66) {};
\end{scope}

\begin{scope}[very thick]
\draw[bend left=10] (21) to  (1);
\draw (18) --  (1);
\draw (13) --  (1);
\draw  (1) --  (2);
\draw  (2) --  (3);
\draw  (1) -- (12);
\draw  (1) -- (17);
\draw[bend left=10]  (1) to (21);
\draw  (3) --  (1);
\draw  (9) --  (3);

\draw  (2) to[bend left=10] (4);
\draw  (2) to[bend right=10] (4);

\draw  (3) to[bend left=10] (5);
\draw  (3) to[bend right=10] (5);

\draw  (5) to[bend left=10]  (6);
\draw  (5) to[bend right=10] (6);

\draw  (3) --  (7);
\draw  (7) --  (8);
\draw  (8) --  (9);
\draw  (7) -- (10);
\draw (10) -- (11);
\draw (11) --  (7);
\draw (14) to[bend left=10] (12);
\draw (12) to[bend left=10] (14);
\draw (16) -- (13);
\draw (12) -- (13);
\draw (13) -- (15);
\draw (15) -- (16);
\draw (17) -- (18);

\draw (17) to[bend left=5] (19);
\draw (19) to[bend left=5] (17);

\draw (17) to[bend left=5] (20);
\draw (20) to[bend left=5] (17);
\draw (24) -- (21);
\draw (21) -- (22);
\draw (22) -- (23);
\draw (23) -- (24);
\end{scope}

\begin{scope}[lks]
\draw (14) to (15);
\draw (16) to[bend right=20] (15);
\draw (4) to[bend right=20] (6);
\draw (14) to[bend left] (24);
\draw (22) to[bend left=10] (20);
\draw (23) to[bend left=10] (18);
\draw (10) to[bend right=10] (9);
\draw (10) to[bend right=10] (8);
\draw (19) to[bend right] (20);
\draw (19) to[bend left] (11);
\draw (8) to (19);
\draw (4) to[bend right] (15);
\draw (6) to[bend right] (11);
\draw (18) to[bend right=10] (9);
\draw (24) to[bend left=10] (18);
\draw (4) to[bend left=10] (16);

\end{scope}

\begin{scope}[shift={(27,-4.5)}]%
\def\ll{30mm} %
\def\vs{12mm} %

\node[right] at (0,0) {$G=(V,E)$};
\node[right,blue] at (0,-\vs) {$L\subseteq \binom{V}{2}$};
\end{scope}

\end{tikzpicture}
 \end{center}
\caption{The black solid graph $G=(V,E)$ is a cactus.
Its vertices of degree $2$ are called \emph{leaves} or \emph{terminals} and are depicted as gray squares.
Together with the dashed edges, which represent the links $L$ and only go between leaves, we obtain a Leaf-to-Leaf CacAP instance.
}\label{fig:LLCacAP_example}
\end{figure}
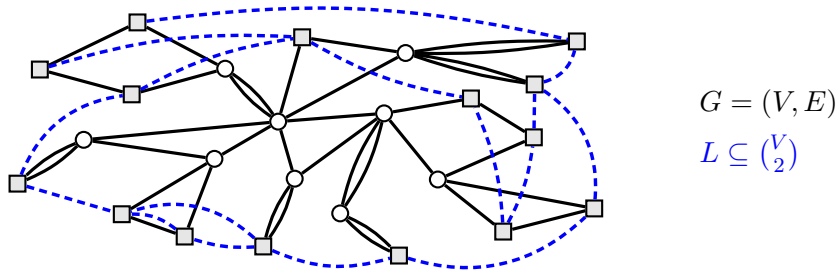
This reduction from CAP to CacAP is approximation preserving and an immediate consequence of the fact that the minimum cuts in a graph can be represented by a cactus (see~\cite{dinitz_1976_structure}).
Vertices of degree $2$ in a cactus are also called \emph{terminals} or \emph{leaves} and the reduction from CAP to CacAP implies that Leaf-to-Leaf CAP instances are transformed into Leaf-to-Leaf CacAP instances.
(See Figure~\ref{fig:LLCacAP_example} for an example of a Leaf-to-Leaf CacAP instance.)
Due to this equivalence, we therefore focus on the more structured Leaf-to-Leaf CacAP instances.

Note that any Leaf-to-Leaf TAP instance $(G=(V,E),L)$ can easily be cast as a Leaf-to-Leaf CacAP instance by adding for each edge $e\in E$ a parallel edge.

\subsection{Our results}\label{sec:ourResults}

Our main result is the following.
\begin{restatable}{theorem}{mainthm}\label{thm:main}
There is a $1.29$-approximation algorithm for Leaf-to-Leaf CAP (and therefore also Leaf-to-Leaf TAP).
\end{restatable}
In the context of leaf-to-leaf instances, this improves on the $1.393$-approximation of~\cite{cecchetto_2021_bridging} (which also works for CAP instances that are not leaf-to-leaf) and also improves on (and is applicable to a much broader set of instances than) the $\frac{4}{3}$-approximation for Leaf-to-Leaf TAP on trees of height $2$ of~\cite{maduel_2010_covering}.

Our main technical contribution, which is the central ingredient of our approach, is an arguably elegant technique to find a good CAP solution by first computing a maximum weight matching over the links with respect to judiciously chosen weights.
This matching is then complemented through a simple LP to an actual CAP solution.
We provide a detailed discussion of our matching-based approach in \cref{sec:matching_approach}.

We highlight that~\cite{maduel_2010_covering}, for the special case of TAP, also used an approach based on first computing a matching and extending it to a solution.
Their approach to extend the matching to a solution uses a credit-based argument, whereas our matching-based approach relies on an LP.

\subsection{Brief overview of main components}\label{sec:overview}

Similar to prior approaches in the field, our matching-based approach provides a guarantee that can be expressed in terms of different link types.
To this end, given a CacAP instance $(G=(V,E),L)$, we can fix an arbitrary root $r\in V$ and define link types with respect to this root as follows.
A link $\{u,v\}\in L$ is called \emph{in-link}, if both $u$ and $v$ lie in the same connected component of $G-r$, where $G-r \coloneqq G[V\setminus \{r\}]$ is the subgraph of $G$ induced by $V\setminus \{r\}$, i.e., the graph obtained from $G$ by removing the root and all edges incident with it.
All other links are called \emph{cross links}.
(See \cref{fig:cross_in_ex}.) We denote the sets of all in-links and cross-links by $L_{\mathrm{in}}$ and $L_{\mathrm{cross}}$, respectively, and for
any link set $U\subseteq L$, we use the shorthands $U_{\mathrm{in}}\coloneqq U\cap L_{\mathrm{in}}$ and $U_{\mathrm{cross}}\coloneqq U \cap L_{\mathrm{cross}}$.
\begin{figure}[!ht]
\begin{center}
\begin{tikzpicture}[%
scale=0.45,
lks/.style={line width=1.3pt, blue, densely dashed},
ns/.style={every node/.append style={thick,draw=black,fill=white,circle,minimum size=6, inner sep=2pt}},
term/.style={rectangle, fill=black!10},
ls/.style={every node/.append style={term}},
crossl/.append style={lks, red, densely dotted},
inl/.append style={lks, blue, densely dashed},
upl/.append style={green!70!black},
]

\pgfdeclarelayer{bg}
\pgfdeclarelayer{fg}
\pgfsetlayers{bg,main,fg}

\begin{scope}

\begin{scope}[ns]
\node  (1) at (18.68,-1.64) {};
\node  (2) at (16.80,-5.50) {};
\node  (3) at (19.08,-5.16) {};
\node[term]  (4) at (16.66,-9.60) {};
\node  (5) at (18.56,-7.16) {};
\node[term]  (6) at (18.24,-8.74) {};
\node  (7) at (20.44,-6.76) {};
\node[term]  (8) at (22.28,-5.68) {};
\node[term]  (9) at (21.00,-3.58) {};
\node[term] (10) at (20.40,-9.36) {};
\node[term] (11) at (22.36,-9.32) {};

\node (12) at (12.24,-4.60) {};
\node (13) at (14.20,-4.64) {};
\node[term] (14) at (11.24,-6.32) {};
\node[term] (15) at (12.68,-7.72) {};
\node[term] (16) at (14.80,-7.42) {};
\node (17) at (24.88,-4.72) {};
\node[term] (18) at (27.60,-4.64) {};
\node[term] (19) at (24.44,-7.64) {};
\node[term] (20) at (26.80,-7.76) {};
\node (21) at (8.04,-4.56) {};
\node[term] (22) at (6.16,-5.36) {};
\node[term] (23) at (6.08,-7.28) {};
\node[term] (24) at (8.00,-6.56) {};
\end{scope}

\node at (1)[above=2pt] {$r$};

\begin{scope}[very thick]

\draw  (1) --  (2);
\draw  (2) --  (3);
\draw  (3) --  (1);
\draw  (9) --  (3);
\draw  (2) to[bend left=15] (4);
\draw  (4) to[bend left=15] (2);
\draw  (3) to[bend left=25] (5);
\draw  (5) to[bend left=25] (3);
\draw  (5) to[bend left=25] (6);
\draw  (6) to[bend left=25] (5);
\draw  (3) --  (7);
\draw  (7) --  (8);
\draw  (8) --  (9);
\draw  (7) -- (10);
\draw (10) -- (11);
\draw (11) --  (7);

\begin{scope}
\draw  (1) to[bend left=2] (21);
\draw  (1) -- (12);
\draw  (1) -- (17);
\draw (21) to[bend left=2]  (1);
\draw (18) --  (1);
\draw (13) --  (1);
\draw (12) to[bend left=15] (14);
\draw (14) to[bend left=15] (12);
\draw (16) -- (13);
\draw (12) -- (13);
\draw (13) -- (15);
\draw (15) -- (16);
\draw (17) -- (18);
\draw (17) to[bend left=13] (19);
\draw (19) to[bend left=13] (17);
\draw (17) to[bend left=13] (20);
\draw (20) to[bend left=13] (17);
\draw (24) -- (21);
\draw (21) -- (22);
\draw (22) -- (23);
\draw (23) -- (24);
\end{scope}

\end{scope}

\end{scope}

\begin{scope}

\begin{scope}[lks]

\begin{scope}[crossl]
\draw (4) to[bend left=20] (24);
\draw (6) to (15);
\draw (8) to (19);
\draw (9) to (18);
\draw (14) to (24);
\draw (11) to (19);
\draw (10) to[bend left=25] (23);
\end{scope}

\begin{scope}[inl]
\draw (4) to[bend right] (11);
\draw (6) to (10);
\draw (8) to[bend right] (9);
\draw (19) to (20);
\draw (14) to (15);
\draw (8) to (11);
\draw (18) to (20);
\draw (15) to[bend left] (16);
\draw (22) to (24);
\end{scope}

\end{scope}
\end{scope}%

\begin{scope}[shift={(29,-3)}]%
\def\ll{30mm} %
\def\vs{12mm} %

\begin{scope}

\coordinate (rl) at (0,0) {};
\coordinate (bl) at (0,-\vs) {};

\draw[crossl] (rl) to ++(1,0) node[right] {cross-links};
\draw[inl] (bl) -- ++(1,0) node[right] {in-links};

\end{scope}
\end{scope}

\end{tikzpicture}
 \end{center}
\caption{A Leaf-to-Leaf CacAP instance with root $r$ on top.
The leaves are drawn as gray squares.
The dashed lines represent the links, which can be partitioned into cross-links (dotted in red) and in-links (dashed in blue).}\label{fig:cross_in_ex}
\end{figure}
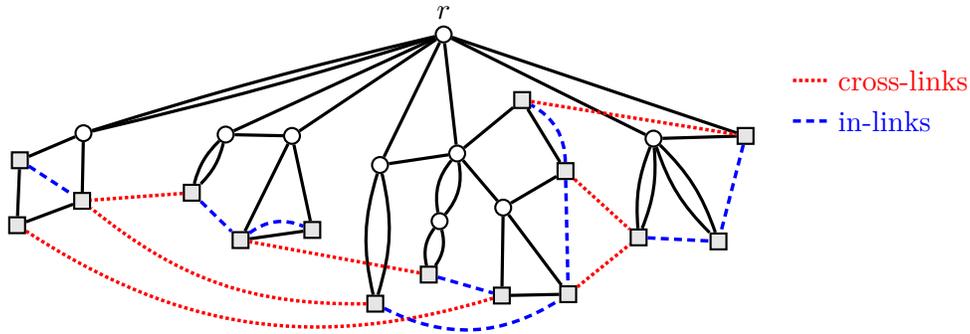

Our matching-based approach leads to a combinatorial procedure that returns a solution with the guarantee stated below.
\begin{restatable}{theorem}{theoreminlinks}\label{thm:1/2inlinks}
For any Leaf-to-Leaf CacAP instance $(G, L)$, we can efficiently compute a solution $F \subseteq L$ such that $|F| \leq |H| + \frac{1}{2}|H_{\mathrm{in}}|$ for any solution $H$ of the instance.
\end{restatable}
In particular, the cardinality of the solution $F$ we return is no larger than the cardinality of any optimal solution $H$ plus half of the number of in-links of $H$.
Clearly, this immediately implies a $\sfrac{3}{2}$-approximation for Leaf-to-Leaf CacAP (and therefore also Leaf-to-Leaf CAP).
The guarantees obtained through \cref{thm:1/2inlinks} are a strengthening, for the leaf-to-leaf case, of guarantees obtainable by prior methods, in particular one developed in~\cite{fiorini_2018_approximating} for TAP and extended in~\cite{cecchetto_2021_bridging} to CacAP, which leads to a guarantee of $|F|\leq |H| + |H_{\mathrm{in}}|$.

Our approach readily leads to better factors than $\sfrac{3}{2}$ when combined with prior techniques that have been developed recently and can be translated to the leaf-to-leaf setting.
These approaches are based on a reduction introduced in~\cite{adjiashvili_2018_beating}, and later extended in~\cite{grandoni_2018_improved,cecchetto_2021_bridging}, which first reduces the given instance to a better structured one, known as $k$-wide (for constant $k$).
\begin{definition}[$k$-wide CacAP]\label{def:k-wide}
Let $k\in \mathbb{Z}_{\geq 1}$.
A CacAP instance $(G,L)$ is \emph{$k$-wide} if there is a vertex $r$ such that each connected component of $G-r$ contains at most $k$ leaves of $G$.
We call $r$ a \emph{$k$-wide root} of $G$, or simply a \emph{root}. Moreover, for each vertex set $W\subseteq V\setminus \{r\}$ of a connected component of $G-r$, we call $G[W\cup \{r\}]$ an $r$-principal subcactus of $G$.
\end{definition}
In particular, the Leaf-to-Leaf CacAP instance highlighted in~\cref{fig:cross_in_ex} is $6$-wide, with respect to the indicated root $r$, and has $4$ principal subcacti.

The following statement shows that, to obtain approximation algorithms for Leaf-to-Leaf CacAP, it suffices to consider $O(1)$-wide instances (up to an arbitrarily small error in the approximation factor).

\begin{theorem}\label{thm:main_reduction}
Let $\alpha \geq 1$ and $\epsilon > 0$. Given an $\alpha$-approximation algorithm $\mathcal{A}$ for any $O(\sfrac{1}{\epsilon^3})$-wide Leaf-to-Leaf CacAP, there is an $\alpha \cdot (1+\epsilon)$-approximation algorithm $\mathcal{B}$ for (unrestricted) Leaf-to-Leaf CacAP that calls $\mathcal{A}$ at most polynomially many times and performs further operations taking polynomial time.
\end{theorem}
\cref{thm:main_reduction} has been proven for (non leaf-to-leaf) CacAP~\cite{cecchetto_2021_bridging}.
As we discuss later, a slight modification of the reduction used in~\cite{cecchetto_2021_bridging} allows for translating this result to the leaf-to-leaf setting, where we want to make sure that the $O(1)$-wide CacAP instance maintains the property of being also leaf-to-leaf if the original instance was leaf-to-leaf.

On $O(1)$-wide CacAP instances (even weighted ones), we can leverage the following result from~\cite{cecchetto_2021_bridging}, which is a consequence of a technique in~\cite{basavaraju_2014_parameterized}.
\begin{lemma}[\cite{cecchetto_2021_bridging}]\label{lem:fpt_in_nb_leaves}
For any weighted $k$-wide CacAP instance $(G=(V,E),L)$ with link costs $c\in\mathbb{R}_{\geq 0}^L$, we can compute in time $3^k \poly(|V|)$ a CacAP solution $F\subseteq L$ with $c(F) \leq c(H) + c(H_{\mathrm{cross}})$ for any solution $H$ of the instance.
\end{lemma}
The above statement is obtained by solving independently and optimally the CacAP problems on each principal subcactus of a given $k$-wide instance and then returning the union of all these solutions.
This is the step that requires $k$ to be constant to be efficiently executable.

Note how the guarantee given by~\cref{lem:fpt_in_nb_leaves} is complementary to the one we obtain through our matching-based procedure as described in \cref{thm:1/2inlinks}.
By returning the better of the two, we immediately obtain a $\sfrac{4}{3}$-approximation for $O(1)$-wide Leaf-to-Leaf CacAP.
\begin{corollary}\label{cor:fourOverThreeKWide}
Let $(G=(V,E),L)$ be an $O(1)$-wide Leaf-to-Leaf CacAP instance.
Computing a solution $F_1\subseteq L$ as claimed by \cref{thm:1/2inlinks} and a solution $F_2\subseteq L$ as claimed by~\cref{lem:fpt_in_nb_leaves} (with $c$ being unit weights), and returning the one of smaller cardinality, leads to a $\sfrac{4}{3}$-approximation.
\end{corollary}
\begin{proof}
Let $\OPT\subseteq L$ be an optimal solution of $(G,L)$.
By \cref{thm:1/2inlinks}, we have $|F_1| \leq |\OPT| + \frac{1}{2}|\OPT_{\mathrm{in}}|$, and \cref{lem:fpt_in_nb_leaves} provides $|F_2|\leq |\OPT| + |\OPT_{\mathrm{cross}}|$.
Hence, the better of the two has size
\begin{equation*}
\min\{|F_1|,|F_2|\} \leq \frac{2}{3} |F_1| + \frac{1}{3} |F_2| \leq |\OPT| + \frac{1}{3} |\OPT_{\mathrm{in}}| + \frac{1}{3} |\OPT_{\mathrm{cross}}| = \frac{4}{3}|\OPT|,
\end{equation*}
as desired.
\end{proof}
Thus, \cref{cor:fourOverThreeKWide} immediately implies, together with \cref{thm:main_reduction}, that there is a $(\sfrac{4}{3}+\epsilon)$-approximation for Leaf-to-Leaf CacAP.
Finally, a recently introduced technique based on stack analysis~\cite{cecchetto_2021_bridging} allows for obtaining approximation factors below $\sfrac{4}{3}$, and implies the claimed $1.29$-approximation for Leaf-to-Leaf CacAP (and therefore also Leaf-to-Leaf CAP) as stated in \cref{thm:main}.

\subsection{Organization of paper}

In \cref{sec:matching_approach}, we introduce our main new technical ingredient, namely a simple matching-based approach to derive Leaf-to-Leaf CacAP (or TAP) solutions, which leads to \cref{thm:1/2inlinks}.
The reduction to $O(1)$-wide Leaf-to-Leaf CacAP instances is discussed in \cref{sec:reduction} 
with some additional explanation of how precisely we reuse results from prior work in \cref{appendix:reduction}.
\cref{sec:stack_analysis} discusses how the stack analysis approach of~\cite{cecchetto_2021_bridging} allows for obtaining approximation factors below $\sfrac{4}{3}$-factor, leading to the claimed $1.29$-approximation for Leaf-to-Leaf CAP.
\section{Our matching-based approach}\label{sec:matching_approach}

We now describe our matching-based approach, which leads to \cref{thm:1/2inlinks}.
In fact, we will show a slight generalization of \cref{thm:1/2inlinks}, which applies to a slightly larger problem class which we call \emph{leaf-to-leaf+ instances}.
This will be useful later on when we combine our matching-based approach with other algorithms to obtain our main result, \cref{thm:main}.
\begin{definition}[Leaf-to-Leaf+ CacAP instance]
A CacAP instance $(G=(V,E),L)$ is a \emph{leaf-to-leaf+ instance} if it has a root $r\in V$ such that every endpoint of a link in $L$ is the root or a leaf of $G$. 
\end{definition}
The main result of this section is the following, which immediately implies \cref{thm:1/2inlinks}.
\begin{theorem}\label{thm:1/2inlinks+}
For any Leaf-to-Leaf+ CacAP instance $(G,L)$, we can efficiently compute a solution $F\subseteq L$ such that $|F|\leq |H| + \frac{1}{2}|H_{\mathrm{in}}|$ for any solution $H$ of the instance.
\end{theorem}

To describe our matching-based approach, consider a Leaf-to-Leaf+ CacAP instance $\Iscr=(G=(V,E),L)$ together with a root $r\in V$.
We denote by $\mathcal{C}$ the set of $2$-cuts of $G$, where, by convention, we only consider cuts not containing $r$, i.e.,
\begin{equation*}
\mathcal{C} \coloneqq \left\{C\subseteq V\setminus \{r\} \colon |\delta_E(C)| = 2\right\}.
\end{equation*}
Hence, the task is to find a smallest link set crossing all sets in $\mathcal{C}$.

\medskip

We will first compute a large matching on the leaves $T$ of $G$ and then show that we can cheaply complete the matching to a solution $F$ with the desired properties.
On a high level, it seems intuitive that good solutions contain large matchings.
In particular, a simple lower bound on the number of links needed in any solution is given by $\sfrac{|T|}{2}$, because each leaf needs to have a link incident with it.
Any instance with an optimal solution close to this lower bound must thus contain a very large matching.
However, simply computing a maximum cardinality matching and then completing it to a solution does not generally lead to strong solutions.
See \cref{fig:badMatching} for an example.
\begin{figure}[!ht]
\begin{center}
\begin{tikzpicture}[%
xscale=0.5,
yscale=0.5,
lks/.style={line width=1.3pt, blue, densely dashed},
ts/.style={every node/.append style={font=\scriptsize}},
ls/.style={every node/.append style={rectangle, fill=black!10}},
lksm/.style={blue, solid},
lkso/.style={line width=3pt,orange!70!black, solid, opacity=0.4},
]

\def\p{6}

\def\hn{3.5}

\def\hl{1.3}

\def\vd{3}

\def\bcac{8pt}

\pgfdeclarelayer{fg}
\pgfdeclarelayer{bg}
\pgfsetlayers{bg,main,fg}

\begin{scope}[every node/.style={thick,draw=black,fill=white,circle,minimum size=6, inner sep=2pt}]

\begin{scope}[ls]
\foreach \x in {1,...,\p} {
\pgfmathsetmacro\xa{\x*\hn-\hl/2}
\pgfmathsetmacro\xb{\x*\hn+\hl/2}
\node (\x a) at (\xa,0) {};
\node (\x b) at (\xb,0) {};
}
\end{scope}

\begin{scope}
\foreach \x in {1,...,\p} {
\pgfmathsetmacro\xx{\x*\hn}
\node (\x) at (\xx,\vd) {};
\node (\x t) at (\xx,2*\vd) {};
}
\end{scope}

\end{scope}

\begin{scope}
\node at (1t)[above=2pt] {$r$};

\node at (2a)[below=2pt] {$u$};
\node at (2b)[below=2pt] {$v$};
\end{scope}

\begin{scope}[very thick]
\pgfmathtruncatemacro\pm{\p-1}
\foreach \t in {1,...,\pm} {
  \pgfmathtruncatemacro\rv{\t+1}
  \draw (\t t) to[bend left = \bcac] (\rv t);
  \draw (\t t) to[bend right = \bcac] (\rv t);
}

\foreach \t in {1,...,\p} {
  \draw (\t) to[bend left = \bcac] (\t t);
  \draw (\t) to[bend right = \bcac] (\t t);
}

\foreach \t in {1,...,\p} {
  \draw (\t a) to[bend left = \bcac] (\t);
  \draw (\t a) to[bend right = \bcac] (\t);
  \draw (\t b) to[bend left = \bcac] (\t);
  \draw (\t b) to[bend right = \bcac] (\t);
}
\end{scope}

\begin{scope}[lks]
\def\bl{50pt}

\begin{scope}[lkso]

\draw (1a)   to[out=-60, in=240, looseness=0.5] (1b);
\draw (\p a) to[out=-60, in=240, looseness=0.5] (\p b);

\pgfmathtruncatemacro\pm{\p-1}
\foreach \u in {1,...,\pm} {
\pgfmathtruncatemacro\v{\u + 1}
\draw (\u b)   to[out=-60, in=240, looseness=0.5] (\v a);
}

\end{scope}

\begin{scope}[lksm]
\foreach \u in {1,...,\p} {
\draw  (\u a) to[bend left = \bl]  (\u b);
}
\end{scope}
\end{scope}

\begin{pgfonlayer}{bg}
\begin{scope}
\filldraw[red,fill opacity=0.2,rounded corners=20pt]
  ($(2a.center)+(-1,-1.0)$) -- ($(2b.center)+(1,-1.0)$) -- node[above right,red,opacity=1] {$C$} ($(2.center)+(0,2)$) -- cycle;
\end{scope}
\end{pgfonlayer}

\coordinate (kn) at ($(\p)+(2.5,0.5)$);
\begin{scope}[shift={(kn)}]%
\def\ll{30mm} %
\def\vs{12mm} %

\draw[lkso] (0,0) -- ++(1,0) node[right,opacity=1] {$\mathrm{OPT}$};
\draw[lksm] (0,-1.0) -- ++(1,0) node[right] {$M$};

\end{scope}

\end{tikzpicture}
 \end{center}
\caption{
A Leaf-to-Leaf TAP instance, represented as a CacAP instance.
The links are shown as blue lines and thick orange lines.
The thick orange links show an optimal solution, which has cardinality $7$.
The blue are the unique maximum cardinality matching $M$ on the link set.
However, complementing the matching $M$ to a solution requires at least $5$ extra links (all orange/$\OPT$ ones except for the left-most and right-most one), leading to a solution of cardinality at least $12$.
By making the example wider, the ratio between the solution obtained by (optimally) complementing $M$ and $\OPT$ approaches $2$.
The blue link $\{u,v\}$ is an example of a bad link, because there is a $2$-cut, highlighted in red, such that each link crossing $C$ has either $u$ or $v$ as one of its endpoints.
}\label{fig:badMatching}
\end{figure}
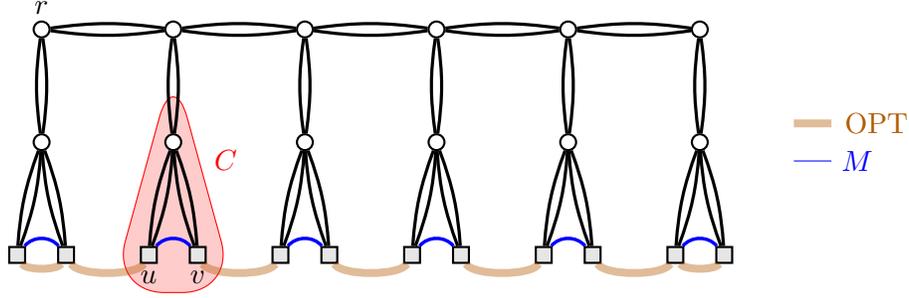

As we formalize in the following, a key reason for a matching not to have a good completion is that it contains a certain type of links, which we call \emph{bad}.
For the special case of TAP, Maduel and Nutov~\cite{maduel_2010_covering} already identified these links as being undesirable and called them \emph{redundant}.
\begin{definition}[Bad link]
For a cut $C\in \Cscr$, let $T_C \subseteq T \cap C$ be the set of leaves in the cut $C$ that are endpoints of links covering the cut $C$. We say that a link $\ell = \{u,v\} \in L$ is a \emph{bad link} if $T_C \subseteq \{u,v\} \subseteq C$ for some $C\in \Cscr$.
\end{definition}
\cref{fig:badMatching} highlights an example of a bad link.

The key lemma in our matching-based approach is the following, which we show in \cref{sec:proof_matching_bound}.
In words, it says that large matchings without bad links can be cheaply augmented to a CacAP solution.
\begin{lemma}\label{lem:matching_bound}
Given a feasible Leaf-to-Leaf+ CacAP instance $(G, L)$ and a matching $M\subseteq L$ on the leaves of $G$ without bad links, we can efficiently find a CacAP solution $F\subseteq L$ with $M\subseteq F$ such that:
\begin{equation}\label{eq:completionGuaranteeM}
|F|\leq |M|+\frac{1}{2}|M_{\mathrm{in}}|+(|T| - 2|M|).
\end{equation}
\end{lemma}
Our matching-based algorithm now simply computes a matching without bad links that minimizes the right-hand side of~\eqref{eq:completionGuaranteeM}, and completes it to a solution with the guarantee claimed by \cref{lem:matching_bound}.
Note that finding a matching without bad links that minimizes
\begin{equation*}
|M|+\frac{1}{2}|M_{\mathrm{in}}| + (|T| - 2|M|) = -|M_{\mathrm{cross}}| - \frac{1}{2}|M_{\mathrm{in}}| + |T|,
\end{equation*}
can easily be done by deleting all bad links and non leaf-to-leaf links and finding a maximum weight matching over the remaining links where each cross-link has a weight of $1$ and each in-link has a weight of $\sfrac{1}{2}$. 

To obtain that our matching-based procedure leads to a solution with the guarantees claimed by \cref{thm:1/2inlinks}, we show that there exist matchings without bad links that lead to a cheap CacAP solution through \cref{lem:matching_bound}.
\begin{lemma}\label{lem:matching_exists}
For any Leaf-to-Leaf+ CacAP instance $(G, L)$, there exists a matching $M\subseteq L$ on the leaves of $G$ without bad links such that 
\begin{equation*}
|M|+\frac{1}{2}|M_{\mathrm{in}}|+(|T| - 2|M|) \leq |H| + \frac{1}{2}|H_{\mathrm{in}}|
\end{equation*}
for any solution $H$ of the instance.
\end{lemma}
Finally, \cref{thm:1/2inlinks} is a straightforward consequence of the above statements.
\begin{proof}[Proof of \cref{thm:1/2inlinks}]
As discussed, we can efficiently find a matching $M$ that minimizes the right-hand side of~\eqref{eq:completionGuaranteeM}, and then use \cref{lem:matching_bound} to efficiently obtain a solution $F$.
\cref{lem:matching_exists} implies that this solution has the desired guarantees.
\end{proof}

We now provide details on the steps performed in our algorithm to obtain a solution $F\supseteq M$ from a leaf-to-leaf matching $M$ without bad links as claimed in \cref{lem:matching_bound}.
To extend $M$ to a CacAP solution $M\cup U$, for some link set $U\subseteq L$, we need that $U$ covers all $2$-cuts not yet covered by $M$. Hence, $M\cup U$ is a CacAP solution if and only if $U\cap \delta_L(C)\neq \emptyset$ for all $C\in\mathcal{C}^M$, where
\begin{equation*}
\mathcal{C}^M \coloneqq \bigl\{ C \subseteq V\setminus\{r\} \ : \ |\delta_E(C)|=2 \text{ and } \delta_L(C)\cap M=\emptyset \bigr\}.
\end{equation*}
To find a good extension $U$, we use an LP based on directed links that is integral, which is an idea introduced in \cite{cecchetto_2021_bridging}.
More precisely, we use the following directed link set, which contains for each original link two antiparallel directed links:
\begin{equation*}
\vec{L}\coloneqq \bigcup_{\{u,v\}\in L}\left\{(u,v),(v,u)\right\},
\end{equation*}
and solve the following LP to find a good extension $U$:
\begin{equation}\label{eq:directed-cut-lp-residual-instance}
\min \left\{ x(\vec{L}) : x \in \mathbb{R}^{\vec{L}}_{\ge 0},\ x\left(\delta^-_{\vec{L}}(C)\right) \ge 1 \quad\forall\; C\in \mathcal{C}^M \right\}.
\tag{dir-LP}
\end{equation}
In words, the LP requires one to ``buy'' directed links, and a cut is only counted as covered if a directed link is entering it.\footnote{We highlight that, in the context of TAP, \eqref{eq:directed-cut-lp-residual-instance} corresponds to the well-known integral up-link LP, which first replaces each link $\{u,v\}$ by at most two other links, one from each endpoint of the link to the vertex of the unique $u$-$v$ path in the tree that lies closest to the root.}
Whereas \eqref{eq:directed-cut-lp-residual-instance} is thus not a relaxation of the problem of finding the best completion $U$ for $M$, it can be shown to be integral with vertices being $\{0,1\}$-vectors (see~\cref{sec:proof_matching_bound}).
Moreover, it has an optimal objective value that is no more than twice as expensive as the optimal completion.
This follows from the observation that one can set, for each link $\{u,v\}$ in an optimal completion $U^*$, the values of $x$ for both $(u,v)$ and $(v,u)$ to $1$ to obtain a feasible solution for \eqref{eq:directed-cut-lp-residual-instance}.

Our matching-based approach is summarized in \cref{algo:new_backbone}.

\begin{algorithm2e}[H]
\addtolength\linewidth{-4ex}
\begin{enumerate}[label=(\arabic*)]\itemsep2pt
\item\label{algitem:computeMatching}
Compute a leaf-to-leaf matching $M\subseteq L$ without bad links maximizing $w(M)$ for 
\begin{equation*}
w(\ell) \coloneqq 
	\begin{cases}
		1 &\text{ if }\ell\text{ is a cross-link} \\
		\tfrac{1}{2} &\text{ if }\ell\text{ is an in-link}.
	\end{cases}
\end{equation*}
\item\label{algitem:complete}
Compute an optimal vertex solution $x^*$ of the LP below ($x^*$ is a $\{0,1\}$-vector):
\begin{equation*}
\renewcommand\arraystretch{1.5}
\begin{array}{r>{\displaystyle}rc>{\displaystyle}ll}
\min & x(\vec{L})                &     &                                         &                              \\
     & x(\delta^-_{\vec{L}}(C))\ & \ge & 1                                       & \forall\; C\in \mathcal{C}^M \\
     & x                         & \in & \mathbb{R}^{\vec{L}}_{\ge 0}  . &
\end{array}
\end{equation*}
Let $U \coloneqq \{\{u,v\}\in L  \colon x^*((u,v))=1 \text{ or } x^*((v,u))=1\}$.
\item Return $M\cup U$.
\end{enumerate}
\caption{Our matching-based approach}\label{algo:new_backbone}
\end{algorithm2e}

In the rest of this section, we provide the details to show that \cref{algo:new_backbone} indeed leads to a Leaf-to-Leaf+ CacAP solution $F=M\cup U$ satisfying $|F| \leq |H| + \frac{1}{2}|H_{\mathrm{in}}|$ for any other solution $H$, which thus implies \cref{thm:1/2inlinks+}.
To this end, we first provide a proof for \cref{lem:matching_exists} in \cref{sec:goodMatchingsExist}, thus showing that a good matching $M$ exists.
We then show in \cref{sec:proof_matching_bound} that our completion procedure used in \cref{algo:new_backbone} leads to a solution with the desired guarantees.

\subsection{Existence of good matching (proof of \texorpdfstring{\cref{lem:matching_exists}}{Lemma 10})}
\label{sec:goodMatchingsExist}

We now provide a proof of \cref{lem:matching_exists}, which shows that there is a good matching $M$.

\begin{proof}[Proof of \cref{lem:matching_exists}]
Let $H\subseteq L$ be any solution of the given Leaf-to-Leaf+ CacAP instance $(G=(V,E),L)$.
Let $M^{H}\subseteq H$ be an inclusion-wise maximal matching consisting only of leaf-to-leaf links in $H$ that are not bad. 
We denote by $T_{\text{cov}}\subseteq T$ the set of leaves covered by $M^{H}$.
Notice that each link $\ell \in H \setminus M^{H}$ satisfies either
    \begin{enumerate}
        \item\label{enum:notbad} one endpoint of $\ell$ belongs to $T_{\text{cov}}$, or
        \item $\ell$ is a bad link. 
    \end{enumerate}
    
Let $H^{\mathrm{bad}}\subseteq H$ be the set of bad links in $H$ for which~\ref{enum:notbad} does not hold. 
    The following claim upper bounds the number of links in $H^{\mathrm{bad}}$.
    \begin{claim}\label{claim:boundoptbad}
    \begin{equation*}
        \sum_{v\in T\setminus T_{\mathrm{cov}}} |\delta_{H}(v)| \geq |T\setminus T_{\mathrm{cov}}| + |H^{\mathrm{bad}}|  .
    \end{equation*}
    \end{claim}
Before proving \cref{claim:boundoptbad}, we show that it implies the desired result.
Assuming \cref{claim:boundoptbad}, we get
\begin{equation}\label{eq:boundOPTInMOPT}
    \begin{aligned}
       |H| &= |M^{H}| + |H \setminus M^{H}|
        \geq |M^{H}| + \sum_{v\in T\setminus T_{\mathrm{cov}}} |\delta_{H}(v)| - |H^{\mathrm{bad}}|
        \geq |M^{H}| + |T\setminus T_{\mathrm{cov}}| \\
        &= |M^{H}| + \left(|T|-2|M^{H}|\right) ,
    \end{aligned} 
\end{equation}
    where the first inequality holds because all links in $H$ that are incident to a vertex in $T\setminus T_{\rm{cov}}$ are contained in $H \setminus M^{H}$ and only links in $H^{\mathrm{bad}}$  have both endpoints in $T\setminus T_{\rm{cov}}$, and the second inequality follows from \cref{claim:boundoptbad}. 
Hence, if $M\subseteq L$ is a matching using only leaf-to-leaf links that are not bad and such that it minimizes $|M| + \tfrac{1}{2}|M_{\mathrm{in}}|+(|T|-2|M|)$ among all such matchings, we obtain
    \begin{equation*}
        |M|+\frac{1}{2}|M_{\mathrm{in}}|+(|T| - 2|M|) \leq|M^{H}|+\frac{1}{2}|M^{H}_{\mathrm{in}}|+\left(|T| - 2|M^{H}|\right) 
        \leq |H| + \frac{1}{2}|H_{\mathrm{in}}| ,
    \end{equation*} 
where the first inequality follows from the fact that $M$ has been chosen to minimize that expression and the second one is due to~\eqref{eq:boundOPTInMOPT} and $\frac{1}{2}|M^{H}_{\mathrm{in}}| \leq  \frac{1}{2}|\mathrm{H}_{\mathrm{in}}|$, which holds because $M^{H}_{\mathrm{in}} \subseteq H_{\mathrm{in}}$.

    It remains to prove \cref{claim:boundoptbad}.
    
    \begin{claimproof}[Proof of \cref{claim:boundoptbad}]
        Let $G^{\mathrm{bad}}= (T\setminus T_{\mathrm{cov}},H^{\mathrm{bad}})$ be the graph induced by links in $H^{\mathrm{bad}}$ on $T\setminus T_{\mathrm{cov}}$. 
        Let $\mathcal{K}=(T_{\mathcal{K}}, H^{\rm{bad}}_{\mathcal{K}})$ be one of its connected components and let $n_{\mathcal{K}}\coloneqq |T_{\mathcal{K}}|$ and $m_{\mathcal{K}}\coloneqq |H^{\rm{bad}}_{\mathcal{K}}|$ denote the number of vertices and edges of $\mathcal{K}$, respectively.
        We will show 
        \begin{equation}\label{eq:bound_for_connected_comp}
        \sum_{v\in T_{\mathcal{K}}} |\delta_{H}(v)| \ge n_{\mathcal{K}} + m_{\mathcal{K}}. 
        \end{equation}
The claim then follows by summing over all connected components of $G^{\rm{bad}}$.

        For each link $\ell \in H^{\rm{bad}}_{\mathcal{K}}$, there exists a cut $C^{\ell}\in \Cscr$ such that $\ell=\{u,v\}$ is a bad link with respect to $C^{\ell}$, i.e., we have $T_{C^{\ell}} \subseteq \{u,v\} \subseteq C^{\ell}$.
         Let
        \begin{equation}\label{eq:defUnionCutCi}
            C^{\mathcal{K}} \coloneqq \bigcup_{\ell \in H^{\rm{bad}}_{\mathcal{K}}} C^{\ell}  .
        \end{equation}
Then $C^{\mathcal{K}}$ is also a $2$-cut in $G$, which can be derived via well-known combinatorial uncrossing arguments as follows.
Start with the family $\mathcal{F}\coloneqq \{C^{\ell} \colon \ell \in H_{\mathcal{K}}^{\mathrm{bad}}\}$.
We maintain that $\mathcal{F}$ is a family of $2$-cuts that are \emph{connected} in the sense that the graph with vertex set $\mathcal{F}$ that has an edge between $C_1\in \mathcal{F}$ and $C_2\in\mathcal{F}$ if $C_1\cap C_2 \neq \emptyset$, is connected.
Because, $\mathcal{K}$ is connected, this holds for the starting $\mathcal{F}$.
We then successively select any two sets $C_1,C_2\in \mathcal{F}$ with $C_1\cap C_2\neq \emptyset$ and replace it by $C_1\cup C_2$.
Because $C_1$ and $C_2$ are $2$-cuts that intersect (and both do not contain the root $r$), also $C_1\cup C_2$ is a $2$-cut.
(This is the step implied by classic uncrossing techniques; see, e.g., Lemma~24 in \cite{cecchetto_2021_bridging}.)
At the end of this procedure we have the single cut $C^{\mathcal{K}}$ in our family, which is therefore also a $2$-cut, as desired.

Because $H$ is a Leaf-to-Leaf+ CacAP solution, there exists a link $\ell_{\mathcal{K}}$ in $H$ that covers the $2$-cut $C^{\mathcal{K}}$.
Note that one endpoint of $\ell_{\mathcal{K}}$ must be contained in $T_{\mathcal{K}}$ because by~\eqref{eq:defUnionCutCi} we have
\begin{equation*}
\delta(C^{\mathcal{K}}) \subseteq \bigcup_{\ell \in H^{\mathrm{bad}}_{\mathcal{K}}} \delta(C^{\ell}),
\end{equation*}
and the choice of $C^{\ell}$ together with the definition of bad links implies that for $\ell=\{v,w\}\in H^{\mathrm{bad}}$, any link in $\delta(C^{\ell})$ must have $u\in T_{\mathcal{K}}$ or $v\in T_{\mathcal{K}}$ as one of its endpoints.
Moreover, one endpoint of $\ell_{\mathcal{K}}$ is not contained in $T_{\mathcal{K}}$, because $T_{\mathcal{K}}\subseteq C^{\mathcal{K}}$.
Hence, in particular, $\ell_{\mathcal{K}}$ is not contained in the set $H^{\rm{bad}}_{\mathcal{K}}$ of links of the component $\mathcal{K}$.
Because each of the $m_{\mathcal{K}}$ many links in $H^{\rm{bad}}_{\mathcal{K}}$ has both endpoints in $T_{\mathcal{K}}$ and the link $\ell_{\mathcal{K}}$ has only one endpoint in $T_{\mathcal{K}}$, we have 
        \begin{equation*}
            \sum_{v\in \mathcal{K}} |\delta_{H}(v)| \geq 2m_{\mathcal{K}} +1 \geq n_{\mathcal{K}} + m_{\mathcal{K}} ,
        \end{equation*}
because $\mathcal{K}$ is connected.
This shows \eqref{eq:bound_for_connected_comp}.
\end{claimproof}
\end{proof}

\subsection{Completing matchings to CacAP solutions}\label{sec:proof_matching_bound}

In this section we show that the completion $U$ of $M$ computed in step~\ref{algitem:complete} of \cref{algo:new_backbone} leads to a solution $F= U\cup M$ fulfilling the guarantees claimed by \cref{lem:matching_bound}.
To this end, we first observe that \eqref{eq:directed-cut-lp-residual-instance} is integral; actually, it only has $\{0,1\}$-vertices, i.e., vertices where each coordinate is either $0$ or $1$.
This guarantees that step~\ref{algitem:complete} of \cref{algo:new_backbone} can indeed be performed as described.

We recall the definition of \textit{residual instance} from \cite[Definition 17]{cecchetto_2021_bridging}, which will be useful when considering instances in which some links have been contracted.

\begin{definition}[residual instance]\label{def:residual_instance}
 Let $\Iscr=(G,L)$ be a CacAP instance and let $L'\subseteq L$. Let $L'=\{\ell_1, \ldots, \ell_h\}$ be a numbering (ordering) of the links in $L'$. The \textit{residual instance} of $\mathcal{I}$ with respect to $L'$ and this numbering is the instance that arises by performing the following contraction operation sequentially for each link $\ell=\ell_1$ up to $\ell=\ell_h$:
contract all vertices that are on every $u$-$v$ path in the cactus, where $u$ and $v$ are the endpoints of $\ell$, into a single vertex.
 \end{definition}
 
As proved in \cite{cecchetto_2021_bridging}, any contraction order leads to the same outcome (Lemma~18 in \cite{cecchetto_2021_bridging}) and hence we will in the following simply talk about the residual instance of $\Iscr$ with respect to $L'$ without specifying an order of the link in $L'$.
Moreover, a residual instance with respect to some link set $L'$ is a CacAP instance whose $2$-cuts correspond precisely to the $2$-cuts in $G$ that have not been covered by $L'$ (Lemma~19 in \cite{cecchetto_2021_bridging}).
This implies that a link set $F$ is a feasible solution for the residual instance of $\Iscr$ with respect to $L'$ if and only if $F\cup L'$ is a feasible solution for the instance $\Iscr$ (Corollary~20 in \cite{cecchetto_2021_bridging}).

\begin{lemma}\label{lem:lp-integral}
All vertices of the feasible region of \eqref{eq:directed-cut-lp-residual-instance} are within $\{0,1\}^{\vec{L}}$.
\end{lemma}
\begin{proof}
We consider the residual instance $(G^M, L^M)$ with respect to the matching $M$.
Then, as proved in \cite[Lemma~19]{cecchetto_2021_bridging}, one can observe that the 2-cuts $\mathcal{C}_{G^M}$
of the residual instance correspond precisely to the cuts in $\mathcal{C}^M$.
By Lemma~14 in \cite{cecchetto_2021_bridging}, the LP 
 \begin{equation*}
\min \left\{ x(\vec{L}) : x \in \mathbb{R}^{\vec{L}}_{\ge 0},\ x\left(\delta^-_{\vec{L}}(C)\right) \ge 1 \text{ for all }C\in \mathcal{C}_{G^M} \right\}
\end{equation*}
is integral and hence also \eqref{eq:directed-cut-lp-residual-instance} is integral.
Moreover, integrality of \eqref{eq:directed-cut-lp-residual-instance} readily implies that all vertices of the feasible region are $\{0,1\}$-vectors.
To observe this, we show that any point $x\in \mathbb{R}^{\vec{L}}_{\geq 0}$ feasible for \eqref{eq:directed-cut-lp-residual-instance} with $x(\ell)>1$ for some $\ell\in \vec{L}$ cannot be a vertex.
Indeed, for such a point $x$, both increasing or decreasing $x(\ell)$ by a small quantity will lead to other feasible points in $\eqref{eq:directed-cut-lp-residual-instance}$.
This implies that $x$ is not an extreme point of the feasible region of \eqref{eq:directed-cut-lp-residual-instance} and therefore not a vertex of it.
\end{proof}

To show that \cref{algo:new_backbone} returns a completion $U$ satisfying $|U|\leq \frac{1}{2}|M_{\mathrm{in}}| + (|T|- 2|M|)$, it remains to show that the optimal value of \eqref{eq:directed-cut-lp-residual-instance} is at most $\frac{1}{2}|M_{\mathrm{in}}| + (|T|-2|M|)$.
We will prove this in two steps.
First, we observe that if we only look at a subset of constraints consisting of a laminar subfamily $\mathcal{L}\subseteq \mathcal{C}^M$, then the statement holds.
Note that this case includes Leaf-to-Leaf+ TAP, where the set of all minimum cuts form a laminar family.
In a second step, we leverage this result on laminar cuts to obtain the desired upper bound on \eqref{eq:directed-cut-lp-residual-instance} also for Leaf-to-Leaf+ CacAP instances, as desired.
\begin{lemma}\label{lem:nonempty-tap-case}
Let $\Lscr \subseteq \mathcal{C}^M$ be a laminar family.
Then the optimum value of the LP
\begin{equation}\label{eq:laminar_polytope}
\min \Bigl\{ x(\vec{L}) : x\in \mathbb{R}^{\vec{L}}_{\ge 0},  x\bigl(\delta^-_{\vec{L}}(C)\bigr) \ge 1 \text{ for all }C\in \Lscr \Bigr\}
\end{equation}
is at most $\frac{1}{2}|M_{\mathrm{in}}| + (|T| -2 |M|)$.
\end{lemma}
\begin{proof}
We construct a feasible solution $x^M$ for \eqref{eq:laminar_polytope} with $x^M(L) \le \frac{1}{2}|M_{\mathrm{in}}| + (|T| -2 |M|)$.
For a link $\ell \in \vec{L}$, let 
\[
\mathcal{C}^{\ell} \coloneqq \{ C\in \Lscr : \ell \in \delta^-(C) \}
\]
be the set of cuts in $\Lscr$ that are covered by $\ell$.
For a terminal $t\in T$, we call a link $(s,t)\in \vec{L}$ a \emph{maximal link entering} $t$ if the set $\Cscr^{(s,t)}$ is inclusion-wise maximal among all links entering $t$.
Similarly, for a link $\{v,w\}\in M$, we call $\ell \in \delta^-(\{v,w\})$ a \emph{maximal link entering} $\{v,w\}$ if the set $\Cscr^{\ell}$ is inclusion-wise maximal among all links entering $\{v,w\}$.
See \cref{fig:laminar_cuts} for an example.
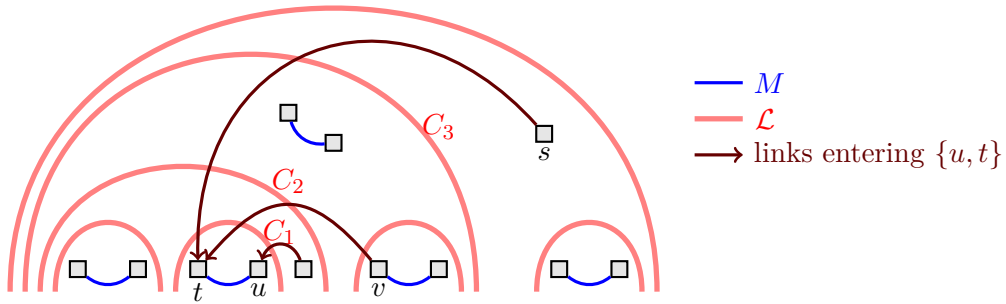
\begin{figure}[!ht]
\begin{center}
\begin{tikzpicture}[scale=0.4]

\path[use as bounding box] (-1, -4) rectangle (28,6);

\tikzset{lksm/.style={blue,solid,very thick}}
\tikzset{ts/.style={rectangle, fill=black!10,thick,draw=black,minimum size=6, inner sep=2pt}}
\tikzset{cuts/.append style={line width=2, red, opacity=0.5}}
\tikzset{links/.append style={red!40!black,very thick, solid,->}}

\begin{scope}
\foreach \n/\type/\x/\y in {%
	1/ts/0/-3,
	2/ts/2/-3,
	3/ts/4/-3,
	4/ts/6/-3,
	5/ts/7.5/-3,
	6/ts/10/-3,
	7/ts/12/-3,
	9/ts/7/2.2,
	10/ts/8.5/1.2,
	11/ts/16/-3,
	12/ts/18/-3,
	13/ts/15.5/1.5%
}{
	\node[\type] (\n) at (\x,\y) {};
}
\end{scope}

\begin{scope}
\node at (3)[below=0.05] {$t$};
\node at (4)[below=0.05] {$u$};
\node at (13)[below=0.05] {$s$};
\node at (6)[below=0.05] {$v$};
\end{scope}

\begin{scope}
\foreach \s/\t in {%
	1/2,
	3/4,
	6/7,
	9/10,
	11/12%
}{
	\draw[lksm] (\s) to[bend right=40] (\t);
}
\end{scope}

\begin{scope}
\foreach \lef/\righ/\loose in {%
	-0.75/2.75/2.25,
	3.25/6.75/2.25,
	9.25/12.75/2.25,
	15.25/18.75/2.25,
	-1.25/8.25/1.5,
	-1.75/13.25/1.8,
	-2.25/19.25/1.5
}{
	\draw[cuts] (\lef,-3.75) to[bend left=90, looseness=\loose] (\righ,-3.75);
}
\end{scope}

\begin{scope}[red]
\node at (6.7,-0.9)[below=0.05] {$C_1$};
\node at (7,0.8)[below=0.05] {$C_2$};
\node at (12,2.6)[below=0.05] {$C_3$};
\end{scope}

\begin{scope}[links]
\draw (6) to[bend right=50, looseness=1.5] (3);
\draw (13) to[bend right=70, looseness=1.5] (3);
\draw (5) to[bend right=70, looseness=1.5] (4);
\end{scope}

\begin{scope}[shift={(20.5,3.2)}]
\def\ll{16mm} %
\def\vs{12mm} %

\draw[lksm] (0,0) -- ++(\ll,0) node[right,opacity=1] {$M$};
\draw[cuts] (0,-\vs) -- (\ll,-\vs) node[right,opacity=1]{$\mathcal{L}$};
\draw[links] (0,-2*\vs) -- (\ll,-2*\vs) node[right,opacity=1]{links entering $\{u,t\}$};
\end{scope}

\end{tikzpicture} \end{center}
\caption{
The link $(s,t)$ is the maximal link entering $t$ because $\Cscr^{(s,t)}=\{C_1,C_2,C_3\} \supseteq \{C_1,C_2\}=\Cscr^{(v,t)}$. Similarly, the link $(s,t)$ is also the maximal link entering $\{t,u\}$.
}\label{fig:laminar_cuts}
\end{figure}

Using that $\Lscr$ is a laminar family, we can observe the following.

\begin{claim}\label{claim:maximal_makes_sense}
If $\ell$ is a link entering a terminal $t\in T$ and $\ell_{\max}$ is a maximal link entering $t$, then $\mathcal{C}^{\ell} \subseteq \mathcal{C}^{\ell_{\max}}$.
Similarly, if $\ell$ is a link entering $\{v,w\}\in M$ and $\ell_{\max}$ is a maximal link entering $\{v,w\}$, then $\mathcal{C}^{\ell} \subseteq \mathcal{C}^{\ell_{\max}}$.
\end{claim}
\begin{claimproof}
Let $\ell$ be a link entering a terminal $t\in T$ and $\ell_{\max}$ a maximal link entering~$t$. 
Because $\mathcal{L}$ is a laminar family, the cuts of $\mathcal{L}$ that contain $t$ form a chain $C_1\subsetneq C_2 \subsetneq \ldots \subsetneq C_q$.
Thus, $\mathcal{C}^\ell$ is a prefix of that chain, i.e., $\mathcal{C}^\ell=\{C_1, C_2, \ldots, C_i\}$ for some $i\in [q]$.
As $\ell_{\max}$ is a maximal link entering~$t$, its prefix must be the largest one among all links, and thus $\mathcal{C}^{\ell_{\max}} \supseteq \mathcal{C}^{\ell}$, as desired.

The second part of the statement follows by an analogous reasoning.
More precisely, let $\{v,w\}\in M$, let $\ell$ be a link entering $\{v,w\}$, and $\ell_{\max}$ be a maximal link entering $\{v,w\}$.
Any cut $C\in \mathcal{L}$ fulfills either $\{v,w\}\subseteq C$ or $\{v,w\}\cap C = \emptyset$, because $\mathcal{L}\subseteq \mathcal{C}^M$ contains only cuts that are not covered by the matching $M$.
Hence, all links in $\mathcal{C}^\ell$ are cuts of $\mathcal{L}$ containing both $v$ and $w$.
As before, the family of all cuts in $\mathcal{L}$ that contain both $v$ and $w$ form a chain because $\mathcal{L}$ is laminar.
Because $\ell_{\max}$ is a maximal link entering $\{v,w\}$, the chain $\mathcal{C}^{\ell_{\max}}$ must be the largest one among all those links.
Thus, $\mathcal{C}^{\ell_{\max}}\supseteq \mathcal{C}^{\ell}$, as claimed.
\end{claimproof}

We now explain how we construct a cheap feasible solution $x^M$ for LP~\eqref{eq:laminar_polytope}.
Let $T^M\subseteq T$ be the set of leaves of $G$ that are not covered by the matching $M$.
Then $|T^M| = |T| -2 |M|$.
We define the vector $x^M$ as follows:
\begin{itemize}
\item for each leaf $t\in T^M$, we choose a maximal link $\ell \in \vec{L}$ entering $t$ and set $x^M_{\ell} \coloneqq 1$;
\item for each in-link $\{v,w\}\in M_{\rm{in}}$, we choose a maximal link $\ell \in \vec{L}$ entering $\{v,w\}$ and set $x^M_{\ell} \coloneqq \tfrac{1}{2}$;
\item for all other links in $\vec{L}$, we set $x^M_{\ell} \coloneqq 0$.
\end{itemize}
Clearly, $x^M(L) = |T^M| + \tfrac{1}{2} |M_{\mathrm{in}}| =  \tfrac{1}{2} |M_{\mathrm{in}}| + (|T| -2 |M|)$.
Thus, it only remains to prove $x^M\bigl(\delta^-_{\vec{L}}(C)\bigr) \ge 1$ for all cuts $C\in \Lscr$.

To this end, fix a cut $C\in \Lscr$.
Recall that $T_C \subseteq T\cap C$ is the set of terminals in $C$ that  have an incident link covering $C$.
We first observe that if $T_C$ contains the head of a cross-link $\ell\in M_{\mathrm{cross}}$, then, because $x^M_\ell=1$ and any cross-links covers all cuts containing its head, we have $x^M(\delta^-_{\vec{L}}(C)) \geq 1$ as desired.

We now distinguish two cases.
First, suppose $T_C$ contains a terminal $t\in T^M$ that is not covered by the matching $M$. 
Then we have $x_{\ell_{\max}}=1$ for a maximal link $\ell_{\max} \in \vec{L}$ entering $t$.
Because $t\in T_C$, there exists a link $\ell \in \vec{L}$ that enters $t$ and covers $C$.
By Claim~\ref{claim:maximal_makes_sense}, the maximality of $\ell_{\max}$ implies that also $\ell_{\max}$ covers $C$.
Hence, $x^M(\delta^-(C)) \ge x^M_{\ell_{\max}} = 1$.

We now consider the remaining case where all terminals in $T^C$ are covered by in-links in the matching $M$.
Let $t_1 \in T^C \subseteq C$. (Note that $T_C$ is not empty because our CacAP instance is feasible.)
In the matching $M$, the terminal $t_1$ is covered by an in-link $\{s_1,t_1\} \in M_{\rm{in}}$. 
Because the link $\{s_1,t_1\}$ is not bad, there exists a terminal $t_2 \in T_C \setminus \{s_1,t_1\}$, which is covered in the matching $M$ by an in-link $\{s_2,t_2\}\in M_{\rm{in}}$.
For each $i\in\{1,2\}$, we chose a maximal link $\ell_i\in\vec{L}$ entering $\{s_i,t_i\}$ and defined $x_{\ell_i}\coloneqq \tfrac{1}{2}$.
Because $t_i\in T_C$, there exists a link $\ell \in \vec{L}$ that enters $t_i$ and covers $C$. 
By Claim~\ref{claim:maximal_makes_sense}, this implies that also the link $\ell_i$ covers the cut $C$.
Hence, we have $x^M(\delta^-(C)) \ge x^M_{\ell_1} + x^M_{\ell_2} = \tfrac{1}{2} + \tfrac{1}{2} = 1$.
 \end{proof}

We now use \cref{lem:nonempty-tap-case}  to show that \eqref{eq:directed-cut-lp-residual-instance} has a cheap feasible solution even if $\mathcal{C}^M$ is not necessarily laminar.
This will complete the proof of \cref{lem:matching_bound}.

\begin{lemma}
The optimum value of \eqref{eq:directed-cut-lp-residual-instance} is at most  $\frac{1}{2}|M_{\mathrm{in}}| + |T| -2 |M|$.
\end{lemma}
\begin{proof}
We consider the dual of \eqref{eq:directed-cut-lp-residual-instance}, which is
\begin{equation}\label{eq:cactus_dual}
\max \left\{\sum_{C\in \mathcal{C}^M} \lambda_C \ : \ \lambda \in \mathbb{R}_{\ge 0}^{\mathcal{C}^M},\ \sum_{C\in \mathcal{C}^M} \lambda_C  \cdot \chi^{\delta^-_{\vec{L}}(C)} \le \chi^{\vec{L}} \right\}.
\end{equation}

Suppose by the sake of deriving a contradiction that the optimum value of the primal LP \eqref{eq:directed-cut-lp-residual-instance} is strictly greater than $\frac{1}{2}|M_{\rm{in}}| + |T| -2 |M|$.
Then, by strong duality, there exists a feasible solution $\lambda^*$ to the dual LP~\eqref{eq:cactus_dual} with 
\begin{equation*}
\sum_{C \in \mathcal{C}^M} \lambda^*_C \ >\  \frac{1}{2}|M_{\mathrm{in}}| + |T| -2 |M|  .
\end{equation*}
We will show by a standard combinatorial uncrossing argument that we may assume the support $\{C\in \mathcal{C}^M : \lambda^*_C > 0\}$ of $\lambda^*$ to be a laminar family $\Lscr\subseteq \mathcal{C}^M$.
Then by restricting $\lambda^*$ to its support, we obtain a feasible solution to 
\begin{equation}\label{eq:tree_dual}
\max \left\{\sum_{C\in \Lscr} \lambda_C \ : \ \lambda \in \mathbb{R}_{\ge 0}^{\Lscr},\ \sum_{C\in \Lscr} \lambda_C  \cdot \chi^{\delta^-_{\vec{L}}(C)} \le \chi^{\vec{L}} \right\}
\end{equation}
with value $\sum_{C \in \Lscr} \lambda^*_C \ >\  \frac{1}{2}|M_{\mathrm{in}}| + |T| -2 |M|$.
This implies, again by duality, that the optimum value of the LP
\begin{equation*}
\min \Bigl\{ x(\vec{L}) : x\in \mathbb{R}^{\vec{L}}_{\ge 0},  x\bigl(\delta^-_{\vec{L}}(C)\bigr) \ge 1 \text{ for all }C\in \Lscr \Bigr\}
\end{equation*}
 is strictly greater than $\frac{1}{2}|M_{\rm{in}}| + |T| -2 |M|$, contradicting Lemma~\ref{lem:nonempty-tap-case}.

It remains to show that we may assume the support of $\lambda^*$ to be laminar.
Let $S, T \in \mathcal{C}^M$ be two crossing sets. Then
\[
 \chi^{\delta^-(S)}+ \chi^{\delta^-(T)} \geq \chi^{\delta^-(S\cap T)}+\chi^{\delta^-(S\cup T)}.
 \]
 Hence, if $S$ and $T$ are in the support of $\lambda^*$, we can decrease $\lambda^*_S$ and $\lambda^*_T$ by some sufficiently small $\epsilon >0$ and at the same time increase $\lambda^*_{S\cap T}$ and $\lambda^*_{S\cup T}$ by $\epsilon$, such that the vector $\lambda^*$ remains feasible for \eqref{eq:cactus_dual} and the value of $\sum_{C \in \mathcal{C}^M} \lambda^*_C $ does not change.
However, the value of $\sum_{C \in \mathcal{C}^M} \lambda^*_C \cdot |C|^2 $ increases.
Therefore, if we choose $\lambda \in \mathbb{R}^{\mathcal{C}^M}_{\ge 0}$ to be a feasible solution to \eqref{eq:cactus_dual} with $\sum_{C \in \mathcal{C}^M} \lambda_C = \sum_{C \in \mathcal{C}^M} \lambda^*_C$ that maximizes $\sum_{C \in \mathcal{C}^M} \lambda_C \cdot |C|^2 $ among all such vectors $\lambda$, then the support of $\lambda$ is a laminar family $\Lscr \subseteq \mathcal{C}^M$.
\end{proof}

\section{Reducing to \texorpdfstring{$\mathbf{O(1)}$}{O(1)}-wide instances}\label{sec:reduction}

Recall that~\cite{cecchetto_2021_bridging} gave a reduction from general instances of CacAP to \emph{$O(1)$-wide instances} (Definition~\ref{def:k-wide}).
In this section we show that this reduction can be adapted for the case of leaf-to-leaf CacAP instances:
we show that, at the expense of an arbitrarily small constant loss in the approximation factor, it suffices to consider $O(1)$-wide leaf-to-leaf CacAP instances.

\begin{theorem}\label{thm:reduction}
Let $\alpha \geq 1$, $\epsilon> 0$, and  $k\coloneqq \frac{64(8+3\epsilon)}{\epsilon^3}$. 
Given an $\alpha$-approximation algorithm $\mathcal{A}$ for $k$-wide Leaf-to-Leaf CacAP instances, there is an $\alpha \cdot (1+2\epsilon)$-approximation algorithm $\mathcal{B}$ for (unrestricted) Leaf-to-Leaf CacAP that calls $\mathcal{A}$ at most polynomially many times and performs further operations taking polynomial time.
\end{theorem}

The reason why we cannot use the reduction to $O(1)$-wide instances given in~\cite{cecchetto_2021_bridging} as a black box is that in our setting we have to make sure that we maintain the property that all links are leaf-to-leaf links. 
The $k$-wide instances that result from the reduction given in~\cite{cecchetto_2021_bridging} are obtained from the original instance by (possibly repeatedly) deleting links and applying the following two types of operations, which we call \emph{splitting} and \emph{contraction}.

\begin{definition}[splitting]\label{def:splitting}
Let $C \subseteq V$ with $|\delta_E(C)|=2$. \emph{Splitting at $C$} leads to two sub-instances $\mathcal{I}_C$ and $\mathcal{I}_{V\setminus C}$, where $\mathcal{I}_C$ is the instance obtained from $\mathcal{I}$ by contracting all vertices except for those in $C$, and $\mathcal{I}_{V\setminus C}$ is the instance obtained from $\mathcal{I}$ by contracting $C$. 
\end{definition}

\begin{definition}[contraction]\label{def:contraction}
 Let $\Iscr=(G=(V,E),L)$ be a CacAP instance and let $\ell\in L$. \emph{To contract $\ell=\{u,v\}$} means that we contract the set of vertices that are contained in every $u$-$v$ path in $G$. 
 \end{definition}

We observe that splitting operations maintain the leaf-to-leaf property.

\begin{lemma}\label{lem:splitting_maintains_leaf_to_leaf}
Let $\Iscr=(G=(V,E),L)$ be an instance of CacAP and let $C \subseteq V$ with $|\delta_E(C)|=2$.
If $\Iscr$ is an instance of Leaf-to-Leaf CacAP, then also the CacAP instances  $\mathcal{I}_C$ and $\mathcal{I}_{V\setminus C}$ that result from splitting $\Iscr$ at $C$ are instances of Leaf-to-Leaf CacAP.

More generally, if $p$  is the number of vertices in $\Iscr$ that are endpoints of a link in $L$, but are not leaves of $G$, then the total number of such non-leaf endpoints of links in $\Iscr_C$ and $\Iscr_{V\setminus C}$ is also $p$.
\end{lemma}
\begin{proof}
Let $p_C$  and $p_{V\setminus C}$ be the number of non-leaf endpoints of links in $\Iscr$ that are contained in $C$ and $V\setminus C$, respectively.
Then $p_C + p_{V\setminus C} = p$.
In the sub-instance $\mathcal{I}_C$, the set $V\setminus C$ is contracted into a single vertex $v_{V\setminus C}$ that is a leaf for the new instance, because $|\delta_E(C)| = |\delta_E(v_{V\setminus C})| =2$. 
For each link $\ell \in \delta_E(C)$, one of its endpoints is replaced by the leaf $v_{V\setminus C}$ while the other endpoint does not change.
As the endpoints of other links are not changed and all leaves of $G$ in $C$ remain leaves after the contraction of $V\setminus C$, the number of non-leaf endpoints of links in $\mathcal{I}_C$ is $p_C$.
By symmetry, it follows that the number of non-leaf endpoints of links in $\mathcal{I}_{V\setminus C}$ is $p_{V\setminus C}$.
\end{proof}

Unfortunately, contractions of links do not maintain the leaf-to-leaf property. See \cref{fig:contraction}.
However, we will use that such contractions are not applied too often.

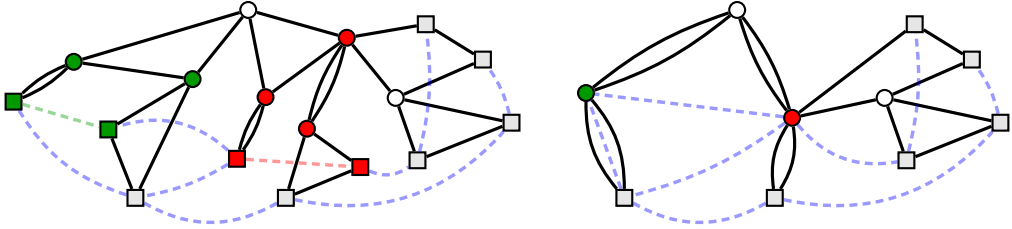
\begin{figure}[h!t]
\begin{center}
\begin{tikzpicture}[%
scale=0.5,
lks/.style={line width=1.3pt, blue, opacity=0.4, densely dashed},
ts/.style={every node/.append style={font=\scriptsize}},
ls/.style={every node/.append style={rectangle, fill=black!10}},
]

\begin{scope}[xshift=-8cm]
\begin{scope}[every node/.style={thick,draw=black,fill=white,circle,minimum size=6, inner sep=2pt}]

\begin{scope}[ls]
\node[fill=red]  (4) at (16.70,-7.96) {};
\node[fill=red]  (6) at (19.98,-8.18) {};
\node (8) at (23.24,-5.32) {};
\node (9) at (21.72,-4.38) {};
\node (10) at (21.50,-8) {};
\node (11) at (24,-7) {};
\node[fill=green!60!black] (14) at (10.76,-6.44) {};
\node[fill=green!60!black] (15) at (13.28,-7.18) {};
\node (16) at (14,-9) {};
\node (17) at (18,-9){};
\end{scope}

\node (1) at (17,-4) {};
\node[fill=red]  (2) at (17.46,-6.32) {};
\node[fill=red]  (3) at (19.62,-4.74) {};
\node[fill=red]  (5) at (18.56,-7.16) {};
\node (7) at (20.92,-6.34) {};
\node[fill=green!60!black] (12) at (12.36,-5.38) {};
\node[fill=green!60!black] (13) at (15.52,-5.84) {};
\end{scope}

\begin{scope}[very thick]

\draw (13) --  (1);
\draw  (1) --  (2);
\draw  (2) --  (3);
\draw  (1) -- (12);
\draw  (3) --  (1);
\draw  (9) --  (3);
\draw  (2) to[bend left=10] (4);
\draw  (2) to[bend right=10] (4);
\draw  (3) to[bend left=10] (5);
\draw  (3) to[bend right=10] (5);
\draw  (5) to (6);
\draw  (5) to (17);
\draw (17) -- (6);
\draw  (3) --  (7);
\draw  (7) --  (8);
\draw  (8) --  (9);
\draw  (7) -- (10);
\draw (10) -- (11);
\draw (11) --  (7);
\draw (14) to[bend left=10] (12);
\draw (12) to[bend left=10] (14);
\draw (16) -- (13);
\draw (12) -- (13);
\draw (13) -- (15);
\draw (15) -- (16);

\end{scope}

\begin{scope}[lks]

\draw[green!60!black] (14) to (15);
\draw (16) to[bend left=20] (14);
\draw[red] (4) to (6);
\draw (10) to[bend right=10] (9);
\draw (11) to[bend right=10] (8);
\draw (4) to[bend right] (15);
\draw (6) to[bend right] (10);
\draw (17) to[bend right] (11);
\draw (4) to[bend left=10] (16);
\draw (17) to[bend left] (16);

\end{scope}
\end{scope}

\begin{scope}[xshift=5cm]
\begin{scope}[every node/.style={thick,draw=black,fill=white,circle,minimum size=6, inner sep=2pt}]

\begin{scope}[ls]

\node (8) at (23.24,-5.32) {};
\node (9) at (21.72,-4.38) {};
\node (10) at (21.50,-8) {};
\node (11) at (24,-7) {};
\node (16) at (14,-9) {};
\node (17) at (18,-9){};

\end{scope}

\node[fill=red]  (4) at (18.46,-6.87) {};
\node[fill=green!60!black] (14) at (12.98,-6.21) {};
\node (1) at (17,-4) {};
\node (7) at (20.92,-6.34) {};
\end{scope}

\begin{scope}[very thick]

\draw (14) to[bend right=10](1);
\draw  (1) to[bend right=10] (4);
\draw  (1) to[bend right=10] (14);
\draw  (4) to[bend right=10] (1);
\draw  (9) --  (4);
\draw  (4) --  (7);
\draw  (7) --  (8);
\draw  (8) --  (9);
\draw  (7) -- (10);
\draw (10) -- (11);
\draw (11) --  (7);
\draw (16) to [bend right=20](14);
\draw (14) to [bend right=20](16);
\draw  (4) to[bend right=20] (17);
\draw (17) to[bend right=20] (4);

\end{scope}

\begin{scope}[lks]

\draw (16) to (14);
\draw (10) to[bend right=10] (9);
\draw (11) to[bend right=10] (8);
\draw (4) to (14);
\draw (4) to[bend left=10] (16);
\draw (17) to[bend left] (16);
\draw (4) to[bend right] (10);
\draw (17) to[bend right] (11);

\end{scope}
\end{scope}

\end{tikzpicture} \end{center}
\caption{
The left picture shows a leaf-to-leaf instance where two links to be contracted are highlighted in green and red, respectively, together with the vertices to be contracted. 
The right picture shows the resulting instance after the two links have been contracted. Notice that after these contraction operations there exist both leaf-to-non-leaf links and non-leaf-to-non-leaf links.
}\label{fig:contraction}
\end{figure}

The following definition formally captures the type of decompositions that we obtain from the reduction from~\cite{cecchetto_2021_bridging}.
 
\begin{definition}\label{def:eps_k_splitting}
Let $\Iscr=(G=(V,E),L)$ be a CacAP instance. A sequence of splitting operations, contraction operations, and deletions of links is called a \emph{ $(\gamma,k)$-splitting of $(G,L)$} if the following two properties are satisfied.
\begin{itemize}
    \item All instances obtained after applying the sequence of operations are $k$-wide; we call these instances the \emph{split-minors} of the $(\gamma,k)$-splitting.
    \item The number of contraction operations is at most $\gamma$.
\end{itemize}
\end{definition}

The proof from \cite{cecchetto_2021_bridging} yields the following general reduction statement.

\begin{theorem}[\cite{cecchetto_2021_bridging}]\label{thm:reduction_with_f}
Let $f: \{ \Iscr : \Iscr \text{ CacAP instance }\} \to \mathbb{R}_{\ge 0}$ be an efficiently computable function and let $\epsilon > 0$.
Suppose we are given an efficient algorithm $\mathcal{A}$ that for any feasible $k$-wide CacAP instance $\Iscr$ computes a solution with at most $\alpha \cdot \OPT(\Iscr)+f(\Iscr)$ many links.
Then there is an algorithm $\mathcal{B}$ that for any feasible CacAP instance $\Iscr$ computes a solution with at most
\begin{equation*}
\alpha \cdot (1+\epsilon)\cdot \OPT(\Iscr)\quad  +\  \max_{\substack{\mathcal{S}\!\text{ an $(\epsilon |\OPT(\Iscr)|,k)$-}\\ \text{splitting of }\Iscr}}\quad \sum_{\mathcal{J}\!\text{ split minor of }\mathcal{S}} f(\mathcal{J})  
\end{equation*}
many links, while calling $\mathcal{A}$ at most polynomially many times and performing further operations taking polynomial time.
\end{theorem}

Because \cite{cecchetto_2021_bridging} does not formally state \cref{thm:reduction_with_f} in this form, we provide details on how the theorem follows from~\cite{cecchetto_2021_bridging} in \cref{appendix:reduction}.
Observe that when $f\equiv 0$, the statement of \cref{thm:reduction_with_f} gives a reduction to $k$-wide instances similar to \cref{thm:reduction}, but for general (not necessarily leaf-to-leaf) CacAP instances. 
In order to prove \cref{thm:reduction}, we will apply \cref{thm:reduction_with_f} with the function $f$  being defined by
\begin{equation}\label{eq:definition_f}
    f\bigl((G=(V,E),L)\bigr) \coloneqq \bigl| \bigl\{ v\in V : v\text{ is not a leaf of $G$, but an endpoint of a link in }L\bigr\}\bigr|  .
\end{equation}

The next lemma shows why this choice of the function $f$ is useful for proving  \cref{thm:reduction}.

\begin{lemma}\label{lem:algorithms_bs}
Let the function $f$ be defined by \eqref{eq:definition_f}. 
Then, for any Leaf-to-Leaf CacAP instance $\Iscr$, we have 
\[
\max_{\substack{\mathcal{S}\!\text{ an $(\gamma,k)$-}\\ \text{splitting of }\Iscr}}\quad \sum_{\mathcal{J}\!\text{ split minor of }\mathcal{S}} f(\mathcal{J}) \ \le\ \gamma .
\]
\end{lemma}
\begin{proof}
Contracting a single link increases the number of non-leaf endpoints of links by at most one.
Moreover, the deletion of links does not increase the total number of non-leaf endpoints and by \cref{lem:splitting_maintains_leaf_to_leaf} the same holds for splitting operations.
Thus, for every $(\gamma, k)$-splitting $\mathcal{S}$, the total number  $\sum_{\mathcal{J}\!\text{ split minor of }\mathcal{S}} f(\mathcal{J})$ of non-leaf endpoints is upper bounded by the number of contraction operations in the splitting $\mathcal{S}$, which is at most $\gamma$ by \cref{def:eps_k_splitting}.
\end{proof}

In order to be able to prove \cref{thm:reduction} using \cref{thm:reduction_with_f}, we will show that using an approximation algorithm for Leaf-to-Leaf CacAP, we can give an approximation algorithm for CacAP with few non-leaf endpoints of links.
More precisely, we show the following lemma.

\begin{restatable}{lemma}{algoflemma}\label{lem:algorithms_as}
Let $\alpha \geq 1$ and the function $f$ be defined by \eqref{eq:definition_f}.
Suppose we are given an $\alpha$-approximation algorithm $\mathcal{A}$ for $k$-wide Leaf-to-Leaf CacAP instances.
Then there is an algorithm $\mathcal{A}'$ that for any  feasible $k$-wide CacAP instance $\Iscr$ computes a solution with at most
 $\alpha \cdot \OPT(\Iscr) + f(\Iscr)$ many links, while calling $\mathcal{A}$ at most once and performing further operations taking polynomial time.
\end{restatable}

\begin{proof}
Given a feasible $k$-wide CacAP instance $\Iscr=(G=(V,E),L)$, we will transform it into a $k$-wide leaf-to-leaf instance. 
As a first step, we will compute a set $X$ of at most $f(\Iscr)$ many links that we can contract to make sure that every non-leaf endpoint of a link in $\Iscr$ is merged with a leaf of $G$ or has no incident links anymore.
In a second step, we will then modify the resulting CacAP instance $\Iscr^X$ such that we obtain a $k$-wide leaf-to-leaf instance $\tilde \Iscr^X$ with the following properties:
\begin{enumerate}[label=(\alph*)]
\item\label{item:opt_does_not_increase} $|\OPT(\tilde \Iscr^X)| \le |\OPT(\Iscr)|$, and
\item\label{item:residual_instance_property} for every solution $F$ of $\tilde \Iscr^X$, the link set $F\cup X$ is a feasible solution for the instance $\Iscr$.
\end{enumerate}
Finally, we use the given algorithm $\Ascr$ to compute an $\alpha$-approximate solution $F$ for the $k$-wide Leaf-to-Leaf CacAP instance $\tilde \Iscr^X$.
Then, using $|X|\le f(\Iscr)$ and properties \ref{item:opt_does_not_increase} and  \ref{item:residual_instance_property}, we can conclude that $F\cup X$ is a solution for $\Iscr$ with
\[
|F\cup X| \le \alpha \cdot \left|\OPT\left(\tilde \Iscr^X\right)\right| + |X| \le  \alpha\cdot|\OPT(\Iscr)| + f(\Iscr).
\]

Let us now explain how we choose the set $X\subseteq L$ and construct the instance $\Iscr^X$.
We will choose a set $X\subseteq L$ with $|X| \le f(\Iscr)$; then we define $\Iscr^X$ to be the residual instance of $\Iscr$ with respect to $X$ (see \cref{def:residual_instance}).
It is well known that $F\subseteq L$ is a feasible solution to the residual instance $\Iscr^X$ if and only if $F\cup X$ is a feasible solution to the original instance (see for example Corollary~20 in \cite{cecchetto_2021_bridging}).
We will call a node of the residual instance that arose from the contraction of several vertices of the original instance $\Iscr$, a \emph{supernode}.
We say that the vertices that were contracted to obtain a supernode $s$ \emph{belong to} $s$.

Let $B\coloneqq \left\{ v\in V : v\text{ is not a leaf of $G$, but an endpoint of a link in }L\right\}$. Then $f(\Iscr) = |B|$.
The following claim describes how we choose the set $X$ of links that we will contract.

\begin{claim}\label{claim:construct_X}
We can efficiently find a set $X\subseteq L$ with $|X|\le f(\Iscr)$ such that the residual instance with respect to $X$ has the following properties:
\begin{enumerate}
\item\label{item:bad_vertices_in_supernode} every vertex $v\in B$ belongs to a supernode, and
\item\label{item:supernodes_nice} for every supernode $s$
\begin{itemize}
\item there is a leaf of $G$ that belongs to $s$, or
\item $s$ has no incident link in the residual instance $\Iscr^X$.
\end{itemize}
\end{enumerate}
\end{claim}
\begin{claimproof}
We construct the set $X$ by the following algorithm.
\begin{enumerate}[label=\arabic*.]
\item Initialize $X=\emptyset$.
\item As long as \ref{item:bad_vertices_in_supernode} or \ref{item:supernodes_nice} is not fulfilled, iterate the following:
\begin{itemize}
\item Choose $v$ to be either a vertex in $B$ violating~\ref{item:bad_vertices_in_supernode} or a supernode violating~\ref{item:supernodes_nice}.
\item Choose an arbitrary link $\ell\in \delta(v)$ and add $\ell$ to $X$.
\item Apply the contraction operation for $\ell$ (\cref{def:contraction}).
\end{itemize}
\end{enumerate}
Note that a link $\ell \in \delta(v)$ exists in every iteration of the algorithm.
If $v\in B$, this follows from the definition of $B$, and if $v$ is a supernode, this follows from the fact that $v$ violates~\ref{item:supernodes_nice}.
At the end of the algorithm, the instance $\Iscr^X$ fulfills \ref{item:bad_vertices_in_supernode} and \ref{item:supernodes_nice} by construction.
In order to prove $|X|\le f(\Iscr) = |B|$, we show that the number of vertices violating \ref{item:bad_vertices_in_supernode} or \ref{item:supernodes_nice} strictly decreases in every iteration of the algorithm.
This will conclude the proof because at the beginning of the algorithm \ref{item:bad_vertices_in_supernode} and \ref{item:supernodes_nice} are violated only by the vertices in $B$ (as no supernodes exist).
Now consider an iteration in which we chose a vertex $v$ violating \ref{item:bad_vertices_in_supernode} or \ref{item:supernodes_nice} and contracted $\ell=\{v,w\}$.
If $w$ is a leaf of $G$ or a supernode with a leaf of $G$ belonging to $w$, then the supernode arising from the contraction of $\ell$ does not violate \ref{item:supernodes_nice}.
Otherwise, by the definition of $B$, the vertex $w$ must be either an element of $B$ (violating~\ref{item:bad_vertices_in_supernode}) or a supernode violating~\ref{item:supernodes_nice}.
When contracting $\ell$, the two vertices $v$ and $w$ violating   \ref{item:bad_vertices_in_supernode} or \ref{item:supernodes_nice} get merged into a single supernode.
In any of these cases the number of vertices violating \ref{item:bad_vertices_in_supernode} or \ref{item:supernodes_nice} decreases strictly.
\end{claimproof}

It remains to show that we can transform the residual instance $\Iscr^X$ into a $k$-wide instance of Leaf-to-Leaf CacAP with properties \ref{item:opt_does_not_increase} and \ref{item:residual_instance_property}.
We construct the instance $\tilde \Iscr^X$ from $\Iscr^X$ as follows.
For every supernode $s$ that has at least one incident link in $\Iscr^X$,
we add a new auxiliary vertex $t_s$ and two copies of the edge $\{s,t_s\}$ to the cactus. 
Then we still have a cactus and $t_s$ is one of its leaves.
Moreover, we replace every link $\{s, v\}$ by the link $\{ t_s, v\}$.
The leaf-to-leaf instance that results from applying this transformation for all supernodes with at least one incident link is the instance $\tilde \Iscr^X = (\tilde G, \tilde L)$.
We view $\tilde L$ as a subset of $L$ by identifying each link in $\tilde L$ with the corresponding link in $L$ from which it arose in the construction.
We now show that $\tilde \Iscr^X$ has the desired properties.

\begin{claim}\label{claim:residual_property}
For every solution $F$ of $\tilde \Iscr^X$, the link set $F\cup X$ is a feasible solution for the instance $\Iscr$.
\end{claim}
\begin{claimproof}
Let $(G^X=(V^X,E^X), L^X) \coloneqq \Iscr^X$ and let $F$ be a solution of $\tilde \Iscr^X = (\tilde G = (\tilde V, \tilde E), \tilde L)$.
Then the graph $(\tilde V, \tilde E \cup F)$ is 3-edge-connected.
Because $(V^X, E^X \cup F)$ arises from this graph by contracting the pair $\{s, t_s\}$ for every supernode $s \in V^X$ for which we added a vertex $t_s$, also the graph $(V^X, E^X \cup F)$ is 3-edge connected, i.e., $F$ is a feasible solution for $\Iscr^X$.
Because $\Iscr^X$ is the residual instance of $\Iscr$ with respect to $X$, we indeed have that $F\cup X$ is a solution for $\Iscr$ (see for example \cite[Corollary 20]{cecchetto_2021_bridging}).
\end{claimproof}

\begin{claim}\label{claim:cost_bound}
The instance $\tilde \Iscr^X$ is feasible and we have $|\OPT(\tilde \Iscr^X)| \le |\OPT(\Iscr)|$.
\end{claim}
\begin{claimproof}
We construct from $\OPT(\Iscr)$ a solution $F$ for $\tilde \Iscr^X$ of at most the same cardinality.
The solution $F$ that we construct contains for every link $\ell=\{v,w\} \in \OPT(\Iscr)$
\begin{enumerate}
\item\label{item:directly_mapped} the link $\ell$ if $v$ and $w$ do not belong to the same supernode, and
\item\label{item:indirectly_mapped} an arbitrary link incident to $t_s$ if $v$ and $w$ belong to the same supernode $s$ and $s$ has an incident link in $\Iscr^X$.
(Note that in this case $t_s$ exists and has an incident link in $\tilde \Iscr^X$.)
\end{enumerate}
By construction $|F|\le |\OPT(\Iscr)|$.
Because $\OPT(\Iscr)$ is a feasible solution for $\Iscr$ and hence for $\Iscr^X$, every violated cut, i.e., every 2-cut of the cactus $\tilde G$ that does not contain a link from $F$, must be of the form $\delta(t_s)$ for some supernode $s$.

Suppose such a violated cut exists. 
The construction of $\tilde \Iscr^X$ implies that there exists a link that is incident to $t_s$.
Moreover, by \cref{claim:construct_X} this implies that there is a leaf $v$ of $G$ that belongs to the supernode $s$. 
The solution $\OPT(\Iscr)$ must contain a link $\ell$ that is incident to $v$.
If both endpoints of $\ell$ belong to the supernode $s$, then $F$ contains a link in $\delta(t_s)$ by \ref{item:indirectly_mapped}.
Otherwise, $\ell \in F \cap \delta(t_s)$, because we replaced the endpoint $s$ of $\ell$ by $t_s$ in the construction of $\Iscr^X$.
Thus, $\delta(t_s)$ is not a violated cut, a contradiction.
\end{claimproof}

\begin{claim}\label{claim:k_wide}
The instance $\tilde \Iscr^X$ is $k$-wide.
\end{claim}
\begin{claimproof}
Let $r$ be a $k$-wide center of $G$ and let $\tilde r$ be either $r$ or the supernode in $\tilde G$ to which $r$ belongs.
We show that $\tilde r$ is a $k$-wide center for $\tilde G$.
To this end, let us consider the vertex set $\tilde W$ of a connected component of $\tilde G - \tilde r$.
Then the vertices in $V$ that are contained in $\tilde W$ or belong to a supernode contained in $\tilde W$ are all part of the same connected component of $G - r$.
Hence, at most $k$ of these vertices are leaves of $G$.
To complete the proof, we will show that every leaf of $\tilde G$ in $\tilde W$ is either a leaf of $G$ or an auxiliary vertex $t_s$ for a supernode $s$ with a leaf of $G$ belonging to $s$.

Let $\tilde t \in \tilde W$ be a leaf of $\tilde G$.
If $\tilde t$  is a vertex in $V$, i.e., it is neither a supernode nor an auxiliary vertex, then its degree in $\tilde G$ is the same as in $G$, implying that $\tilde t$ is a leaf of $G$.
If $\tilde t$ is an auxiliary vertex $t_s$ for a supernode $s$, then by \cref{claim:construct_X}, there exists a leaf of $G$ that belongs to $s$.
Finally, suppose $\tilde t$ is a supernode. 
Then $\tilde t$ cannot have any incident links as otherwise we would have added the leaf $t_s$ incident to $\tilde t$, However, because $\tilde t$ is a leaf of $\tilde G$, this contradicts the feasibility of the instance $\tilde \Iscr^X$ (\cref{claim:cost_bound}).
\end{claimproof}

\cref{claim:residual_property}, \cref{claim:cost_bound}, and \cref{claim:k_wide} imply that $\tilde \Iscr^X$ indeed has the desired properties.
\end{proof}

\cref{thm:reduction} now follows directly from \cref{thm:reduction_with_f}, \cref{lem:algorithms_bs}, and \cref{lem:algorithms_as}.

\newcommand{\etalchar}[1]{$^{#1}$}

\appendix
\section{Details of the reduction to \texorpdfstring{$\mathbf{O(1)}$}{O(1)}-wide instances from
\texorpdfstring{\cite{cecchetto_2021_bridging}}{[3]}}\label{appendix:reduction}

In this Section we provide details on how \cref{thm:reduction_with_f} follows from \cite{cecchetto_2021_bridging}.
Theorem~5 in \cite{cecchetto_2021_bridging} is the same as \cref{thm:reduction_with_f} with $f\equiv 0$.
The same proof can be extended to the case of a general efficiently computable function $f$.
We will first provide a brief outline of the overall proof and then provide adapted versions of those lemmas and proofs from \cite{cecchetto_2021_bridging} that require (small) changes in order to obtain Theorem~\ref{thm:reduction_with_f}.

Given an instance $\Iscr=(G,L)$, let $P_{\mathrm{CacAP}}(\mathcal{I})$ be the convex hull of all actual solutions, namely
\begin{equation*}
P_{\mathrm{CacAP}}(\mathcal{I}) \coloneqq \text{conv}(\{\chi^F : F\subseteq L, (V,E\cup F) \text{ is $3$-edge-connected}\})\ .
\end{equation*}
The proof of Theorem~\ref{thm:reduction_with_f}  is based on the round-or-cut framework.
We will show that there exists an efficient algorithm that, given a point $x\in [0,1]^L$, returns either
\begin{enumerate}
\item\label{item:round} a CacAP solution with at most 
\begin{equation*}
\alpha \cdot (1+\epsilon)\cdot \OPT(\Iscr)\quad  +\  \max_{\substack{\mathcal{S}\!\text{ an $(\epsilon \cdot x(L),k)$-}\\ \text{splitting of }\Iscr}} \quad \sum_{\mathcal{J}\!\text{ split minor of }\mathcal{S}} f(\mathcal{J})  
\end{equation*}
many links, or
\item\label{item:seperate} a vector $w\in \mathbb{R}^L$ such that $w^T x < w^T x'$ for all $x' \in P_{\mathrm{CacAP}}(G,L)$.
\end{enumerate}
Having such an algorithm, we guess the number $|\OPT(\Iscr)|$ of links in an optimum solution and run the Ellipsoid method to find a point in the polytope $\{x \in P_{\mathrm{CacAP}}(\mathcal{I}) : x(L) = |\OPT(\Iscr)| \}$.
As a separation oracle, we can return a separating hyperplane if the given point $x\in [0,1]^L$ does not fulfill  $x(L)= |\OPT(\Iscr)|$ and otherwise apply the algorithm that returns either a CacAP solution as in \ref{item:round} or a separating hyperplane as in \ref{item:seperate}.
If we obtain a CacAP solution as in \ref{item:round}, this solution has the desired properties (as required in Theorem~\ref{thm:reduction_with_f}) because $x(L)= |\OPT(\Iscr)|$; in this case we are done and we terminate the Ellipsoid method.
Otherwise, we have obtained the separating hyperplane that we need to continue with the Ellipsoid method.
Because the Ellipsoid terminates in polynomial time, it indeed suffices to provide an algorithm that given a point $x\in [0,1]^L$ either returns~\ref{item:round} or~\ref{item:seperate}.
Polynomial-time termination of the Ellipsoid Method follows by classical results (see~\cite[Theorem~(6.4.1)]{groetschel_1993_geometric}) because 
the polytope $\{x \in P_{\mathrm{CacAP}}(\mathcal{I}) : x(L) = |\OPT(\Iscr)| \}$, over which we run the Ellipsoid Method, has facet complexity that is bounded polynomially in $|L|$, because it is a $\{0,1\}$-polytope in $[0,1]^L$.

We now give a brief outline of the algorithm used to round $x\in [0,1]^L$ or separate $x$ from $P_{\mathrm{CacAP}}(G,L)$.
The algorithm proceeds in three steps:
\begin{itemize}
\item The first step, called \emph{heavy cut covering}, computes a set $L_H \subseteq L$ of links with $|L_H|\le \frac{\epsilon}{2}x(L)$ that covers all \emph{$x$-heavy cuts}, i.e., we have $L_H \cap \delta(C) \ne \emptyset$ for all 2-cuts $C\in \mathcal{C}$ with $x(\delta_L(C)) > \tfrac{16}{\epsilon}$.
Then we consider the residual instance with respect to $L_H$.
\cite[Lemma~22]{cecchetto_2021_bridging} shows that this residual instance has no heavy cuts.
Moreover, if we take any solution $F$ for the residual instance, then $F\cup L_H$ is a feasible solution for $\Iscr$ (by \cite[Corollary~20]{cecchetto_2021_bridging}).
\item In the second step, we split the instance into $k$-wide instances using the splitting operations from \cref{def:splitting}.
\item Finally, for every split minor arising from the splitting procedure we either find a CacAP solution or we obtain a separating hyperplane.
If we find a CacAP solution for all of the split minors, we merge them to a CacAP solution of the overall instance.
\end{itemize}

Let us now recap the main statement from \cite{cecchetto_2021_bridging} that we use for merging the CacAP solutions for the split minors to a solution of the overall instance.
Recall that splitting an instance $\Iscr$ at a cut $C\in \mathcal{C}$ leads to two sub-instances $\mathcal{I}_C$ and $\mathcal{I}_{V\setminus C}$, where $\mathcal{I}_C$ is the instance obtained from $\mathcal{I}$ by contracting all vertices except for those in $C$, and $\mathcal{I}_{V\setminus C}$ is the instance obtained from $\mathcal{I}$ by contracting $C$.
If we want to merge solutions $F_{C}$, $F_{V\setminus C} \subseteq L$ to the sub-instances $\mathcal{I}_C$ and $\mathcal{I}_{V\setminus C}$, respectively, to a solution of $\Iscr$, it may be not enough to simply take the union of $F_C$ and $F_{V\setminus C}$. 
However, we can bound the number of links that need to be added to $F_C\cup F_{V\setminus C}$ to obtain a solution to $\Iscr$ as shown in~\cite{cecchetto_2021_bridging}.
\begin{proposition}[Proposition 11, \cite{cecchetto_2021_bridging}]\label{prop:combine_split_sol}
Given a feasible CacAP instance $\mathcal{I}=(G=(V,E),L)$, a $2$-cut $C\in \mathcal{C}$, and solutions $F_C, F_{V\setminus C} \subseteq L$ to $\mathcal{I}_C$ and $\mathcal{I}_{V\setminus C}$, respectively, one can efficiently compute a link set $F\subseteq L$ such that
\begin{enumerate}
\item $F_C \cup F_{V\setminus C} \cup F$ is a CacAP solution to $\mathcal{I}$, and
\item $|F|\leq |\delta_L(C)\cap F_C|-1$.
\end{enumerate}
\end{proposition}

We now explain in detail those parts of the proof from \cite{cecchetto_2021_bridging} that need to be adapted to obtain \cref{thm:reduction_with_f}.
First, we provide a generalization of Lemma~32 from \cite{cecchetto_2021_bridging}, which yields an algorithm for $k$-wide instances that either returns a cheaply mergeable solution or a separating hyperplane.
In this algorithm we will use the fact that the function $f$ is efficiently computable.
We need the following definition.

\begin{definition}[$\lambda$-contraction minor]
For an instance $\mathcal{I}=(G,L)$, a \emph{$\lambda$-contraction minor} is an instance $\mathcal{I}'$ that can be obtained by $\mathcal{I}$ by
\begin{itemize}
    \item deleting links, and
    \item applying contraction operations (see \cref{def:contraction}) for at most $\lambda$ many links.
\end{itemize}
\end{definition}

Now we can generalize Lemma~32 from \cite{cecchetto_2021_bridging} as follows.\footnote{Lemma~32 from \cite{cecchetto_2021_bridging} applies to so-called unsplittable instances, but these instances are $k$-wide by \cite[Lemma~12]{cecchetto_2021_bridging} and hence Lemma~\ref{lem:algo_unsplittable} is indeed a generalization of Lemma~32 from \cite{cecchetto_2021_bridging}.
We also remark that the reason why we bound $|F| + |\delta_{F}(s)|$ instead of $|F|$ in Lemma~\ref{lem:algo_unsplittable} is that this is needed to guarantee that $F$ is cheaply mergeable.
}

\begin{lemma}\label{lem:algo_unsplittable}
 Let  $\epsilon \ge 0$ and $\epsilon' = \frac{\epsilon}{4}$.
 Suppose there is an efficient algorithm $\Ascr$ that for any $k$-wide CacAP instance $\Jscr$ returns a solution of cost at most $\alpha\cdot \OPT(\Jscr) + f(\Jscr)$, where $f$ is a function that is efficiently computable.
 Then there is a polynomial-time algorithm that, given a $k$-wide CacAP instance $\mathcal{I}=(G=(V,E),L)$, a vector $x\in [0,1]^L$, and a vertex $s$ of $G$ with  $|\delta_E(s)|=2$ and $x(\delta_L(s))\le \frac{16}{\epsilon}$, either returns
 \begin{itemize}
  \item a CacAP solution $F\subseteq L$ with 
  \begin{equation*}
  \begin{aligned}
      |F| + |\delta_{F}(s)| \leq \  & (1+\epsilon')\cdot \alpha \cdot\left( x(L) + x(\delta_L(s))\right)\\ &+ \max\{f(\Iscr') \colon \Iscr' \text{ is a $\lambda$-contraction minor of $\Iscr$}\} ,
  \end{aligned}
  \end{equation*}
  where $\lambda= \frac{1+\epsilon'}{\epsilon'}\cdot \frac{16}{\epsilon}$, or
  \item a vector $w\in \mathbb{R}^L$ such that $w^T x < w^T x'$ for all $x' \in P_{\mathrm{CacAP}}(\Iscr)$.
 \end{itemize}
\end{lemma}
\begin{proof}
The proof of this lemma follows the proof of Lemma~32 in~\cite{cecchetto_2021_bridging}.
Let $r\in V$ be a $k$-wide root.
If $x(\delta_L(s))<1$, we have $x(\delta_L(s)) < x'(\delta_L(s))$ for all $x'\in P_{\mathrm{CacAP}}(\Iscr)$ and can thus return $w=\chi^{\delta_L(s)}$.
Otherwise, we proceed as follows.
For every set $S \subseteq \delta_L(s)$ with $|S| \le \lambda$, we apply the given algorithm $\Ascr$ for $k$-wide instances to the residual instance $\Iscr'$ of $(G, L\setminus \delta_L(s))$ with respect to $S$. 
This instance $\Iscr'$ arises from $(G,L)$ by applying contraction operations for the links in $S$ after deleting all links in $\delta_L(s)$. 
Thus, $\Iscr'$ is a $\lambda$-contraction minor and it is $k$-wide by~\cite[Lemma~31]{cecchetto_2021_bridging}.

If for some such set $S$, the algorithm $\Ascr$ returns a solution $F'$ with $|F'| + 2 |S| \le  (1+\epsilon')\cdot \alpha \cdot\left( x(L) + x(\delta_L(s))\right) + f(\Iscr')$, we return $F=F' \cup S$. Here we use the fact that $f$ is an efficiently computable function; hence, we can efficiently check whether we are in this case.
Otherwise, we define $\mu:= 1 +  \frac{\epsilon' ( x(L) + x(\delta_L(s)) ) }{ x(\delta_L(s)) }$ and claim that
\begin{equation}
x'(L)+ \mu \cdot x'(\delta_L(s)) > x(L) + \mu \cdot x(\delta_L(s)) \quad \forall\; x'\in P_{\mathrm{CacAP}}(\Iscr),\label{eq:sep_hyper_when_A_fails}
\end{equation}
which leads to a vector $w\in \mathbb{R}^L$ as desired.
Suppose that~\eqref{eq:sep_hyper_when_A_fails} does not hold. 
Then there exists a solution $F^*$ of $\Iscr$ with
\begin{equation}\label{eq:solution_violating_cut}
|F^*| + \mu \cdot | \delta_{F^*}(s) | 
 \le x(L) + \mu \cdot x(\delta_L(s))  = (1+\epsilon') \cdot (x(L) + x(\delta_L(s)).
\end{equation}
For $S = \delta_{F^*}(s)$ this implies
\begin{equation*}
\frac{\epsilon' \left( x(L) + x(\delta_L(s)) \right) }{ x(\delta_L(s)) } \cdot |S|\ \le\ \mu \cdot  |S|\ \le\  (1+\epsilon') \cdot \left( x(L) + x(\delta_L(s)) \right)
\end{equation*}
and hence
\begin{equation*}
|S| \ \le\ \frac{1+\epsilon'}{\epsilon'} x(\delta_L(s))\ \le\ \lambda.
\end{equation*}
Thus, we considered the set $S$ in our algorithm described above.
Let $F'$ be the output of the $\alpha$-approximation algorithm $\Ascr$ applied to the residual instance $\Iscr'$ of $(G, L\setminus \delta_L(s))$ with respect to $S$.
Then $|F'| \le \alpha \cdot |F^* \setminus S|+f(\Iscr')$ because the set $F^*\setminus S$ is a feasible solution of this residual instance (see~\cite[Corollary~20]{cecchetto_2021_bridging}). 
Therefore, using $\alpha \ge 1$,  $S=\delta_{F^*}(s)$, and \eqref{eq:solution_violating_cut}, we get
\begin{equation*}
\begin{aligned}
|F' | + 2 |S| &\le \alpha \cdot |F^*\setminus S|+f(\Iscr') + 2 |S| \\ & \le \alpha \cdot (|F^*| + |S|)+f(\Iscr') \\ &\leq \alpha \cdot (|F^*| + \mu |\delta_{F^*}(s)|) + f(\Iscr')\\
&\le (1+\epsilon')\cdot \alpha \cdot (x(L) + x(\delta_L(s))+f(\Iscr'),
\end{aligned}
\end{equation*}
contradicting the fact that we did not return $F' \cup S$.
\end{proof}

The next lemma provides a generalization of Lemma~33 from \cite{cecchetto_2021_bridging} (except that we choose $k$ slightly differently).
This lemma extends the algorithm that, given $x\in [0,1]^L$, either rounds it to a cheap solution or provides a separating hyperplane from $k$-wide instances to general instances without $x$-heavy cuts.

\begin{lemma}\label{lem:round_and_cut_without_heavy}
Let  $0 \le \epsilon \le 1$ and $k\coloneqq \frac{64(8+3\epsilon)}{\epsilon^3}$.
Suppose there exist an efficient algorithm $\Ascr$ that for any $k$-wide CacAP instance $\Jscr$ returns a solution of cost at most $\alpha\cdot \OPT(\Jscr) + f(\Jscr)$, where $f$ is a function that is efficiently computable.
Then, for any CacAP instance $\mathcal{I}=(G=(V,E),L)$ and $x\in [0,1]^L$ such that no cut is $x$-heavy, there is a polynomial-time algorithm that computes either
\begin{itemize}
\item a CacAP solution $L'$ with
\begin{equation*}
    |L'|\le \alpha \cdot \left(1+\frac{\epsilon}{2}\right) \cdot x(L) +\  \max_{\substack{\mathcal{S}\!\text{ a $(\gamma,k)$-}\\ \text{splitting of }\Iscr}}\quad \sum_{\mathcal{J}\!\text{ split minor of }\mathcal{S}} f(\mathcal{J}) ,
\end{equation*}
where $\gamma=\frac{\epsilon}{4}\cdot |T|$ with $T$ being the set of leaves of $G$, or
\item a vector $w\in \mathbb{R}^L$ such that $w^T x < w^T x'$ for all $x' \in P_{\mathrm{CacAP}}(\mathcal{I})$.
\end{itemize}
\end{lemma}
\begin{proof}
The proof of this lemma follows from Lemma~33 in~\cite{cecchetto_2021_bridging} and proceeds by induction on the number of vertices.
We fix an arbitrary root $r\in V$.
If there is no $2$-cut $C\subsetneq V\setminus \{r\}$ such that $|C\cap T|> \tfrac{k}{2}$, then, following the terminology of~\cite{cecchetto_2021_bridging}, the instances is called \emph{unsplittable}. (More precisely, for an instance to be unsplittable no such cut $C$ must exist that is not $x$-heavy; however, as no $x$-heavy cuts exist in our instance by assumption, this condition can be dropped here.)
By \cite[Lemma 12]{cecchetto_2021_bridging}, every unsplittable instance is $k$-wide, and we can thus apply algorithm $\Ascr$ to $\Iscr$ to obtain a feasible solution $L'$.
If $|L'| \le \alpha \cdot x(L) + f(\Iscr)$, then 
\begin{equation*}
    |L'| \le \alpha \cdot x(L) + f(\Iscr) \le \alpha \cdot \left(1+\frac{\epsilon}{2}\right) \cdot x(L) +\  \max_{\substack{\mathcal{S}\!\text{ a $(\gamma,k)$-}\\ \text{splitting of }\Iscr}}\quad \sum_{\mathcal{J}\!\text{ split minor of }\mathcal{S}} f(\mathcal{J}),
\end{equation*}
because $\Iscr$ itself is a $(\gamma,k)$-split-minor. In this case we return $L'$. 
Otherwise, we return $w = \chi^{L}$, which fulfills $w^T x < w^T x'$ for all $x' \in P_{\mathrm{CacAP}}(\mathcal{I})$ because $\Ascr$ returns a solution of cost at most $\alpha\cdot \OPT(\Jscr) + f(\Jscr)$.

Hence, we may assume that $ \Iscr$ is splittable at some cut $C$, i.e., there is a cut $C$ that fulfills $|T\cap C| > \tfrac{k}{2}$ and $C\subsetneq V\setminus \{r\}$. We choose $C\in \mathcal{C}$ to be a minimal $2$-cut with these properties.

We denote by $\mathcal{I}_C=(G_C, L_C)$ the instance arising from $\mathcal{I}$ by contracting $V\setminus C$ and we choose as new root the vertex $s$ arising from the contraction of $V\setminus C$. (Note that $\{r\} \subsetneq V\setminus C$.)
Because $C$ was chosen minimally, the instance $\mathcal{I}_C$ with root $s$ is $k$-wide (this is again a consequence of \cite[Lemma 12]{cecchetto_2021_bridging}).
We apply Lemma~\ref{lem:algo_unsplittable} to $\mathcal{I}_C$ to either obtain a CacAP solution $F_C$ for $\mathcal{I}_C$ with
\begin{equation}\label{eq:F_C}
\begin{aligned}
|F_C| + |\delta_{F_C}(s)| \leq &\left(1+\tfrac{\varepsilon}{4}\right)\cdot \alpha\cdot (x(L_C) + x(\delta_L(s))) \\ &+ \max\{f(\Iscr') \colon \Iscr' \text{ is a $\lambda$-contraction minor for $\Iscr_C$}\},
\end{aligned}
\end{equation}
where $\lambda= \frac{1+\epsilon'}{\epsilon'}\cdot \frac{16}{\epsilon}$,
or a vector $w_C\in \mathbb{R}^{L_C}$ satisfying $w_{C}^T x_{C} < w_{C}^T x'$ for all $x' \in P_{\mathrm{CacAP}}(\mathcal{I}_C)$, where $x_{C}$ is the restriction of the vector $x$ to the links in $L_C$.

Moreover, we denote by $\mathcal{I}_{V\setminus C}=(G_{V\setminus C}, L_{V\setminus C})$ the CacAP instance arising from $\mathcal{I}$ by contracting $C$. By induction, we can obtain a CacAP solution $F_{V\setminus C}$ for $\Iscr_{V\setminus C}$ such that
\begin{equation}\label{eq:F_V-C}
|F_{V\setminus C}|\leq \alpha \cdot\left(1+\tfrac{\varepsilon}{2}\right) \cdot x(L_{V\setminus C}) + \max_{\substack{\mathcal{S}\!\text{ a $(\gamma',k)$-}\\ \text{splitting of }\Iscr_{V\setminus C}}}\quad \sum_{\mathcal{J}\!\text{ split minor of }\mathcal{S}} f(\mathcal{J}),
\end{equation}
where $\gamma'= \tfrac{\epsilon}{4} \cdot |T_{V\setminus C}|$ with $T_{V\setminus C}$ being the set of leaves of $G_{V\setminus C}$, or a vector $w_{\bar{C}}\in \mathbb{R}^{V\setminus C}$ such that $w_{\bar{C}}^T x_{\bar{C}} < w_{\bar{C}}^T x'$ for all $x' \in P_{\mathrm{CacAP}}(\mathcal{I}_{V\setminus C})$, where $x_{\bar{C}}$ is the vector $x$ restricted to the entries corresponding to links in $L_{V\setminus C}$.

If we obtained a vector $w_C\in \mathbb{R}^{L_C}$ such that $w_C^T x_C < w_C^T x'$ for all $x' \in P_{\mathrm{CacAP}}(\mathcal{I}_C)$, we can extend it to a vector $w\in \mathbb{R}^L$ such that $w^T x < w^T x'$ for all $x' \in P_{\mathrm{CacAP}}(\mathcal{I})$ by setting $w_\ell =0$ for all $\ell \in L \setminus L_C$.
Here we use that restricting a CacAP solution $F$ of $\mathcal{I}$ to $F\cap L_C$ yields a CacAP solution for $\mathcal{I}_C$.
We can proceed analogously if we have a vector $w_{\bar{C}}\in \mathbb{R}^{L_{V\setminus C}}$ such that $w_{\bar{C}}^T x < w_{\bar{C}}^T x'$ for all $x' \in P_{\mathrm{CacAP}}(\Iscr_{V\setminus C})$.

Otherwise, we obtained solutions $F_C$ and $F_{V\setminus C}$ as discussed above and apply Proposition~\ref{prop:combine_split_sol}.
This yields a set $F\subseteq L$ such that $F \cup F_C \cup F_{V\setminus C}$ is a CacAP solution for $\Iscr$ and $|F| \le |F_C \cap \delta_L(C)| = |\delta_{F_C}(s)|$.
Thus, combining~\eqref{eq:F_C} and~\eqref{eq:F_V-C}, we obtain

\begingroup
\allowdisplaybreaks
\begin{align}\label{eq:computation_minor}
    |F \cup F_C \cup F_{V\setminus C}|\ \le&\ |F_{V\setminus C}| + |F_C| + |\delta_{F_C}(s)| \notag\\
    \le&\ \alpha \cdot\left(1+\tfrac{\varepsilon}{2}\right) \cdot x(L_{V\setminus C}) + \max_{\substack{\mathcal{S}\!\text{ a $(\gamma',k)$-}\\ \text{splitting of }\Iscr_{V\setminus C}}}\quad \sum_{\mathcal{J}\!\text{ split minor of }\mathcal{S}} f(\mathcal{J}) \notag\\
    &\ + \left(1+\tfrac{\varepsilon}{4}\right)\cdot \alpha\cdot (x(L_C) + x(\delta_L(s))) \notag\\ &\ + \max\{f(\Iscr') \colon \Iscr' \text{ is a $\lambda$-contraction minor for $\Iscr_C$}\} \notag\\[4mm]
    =&\  \alpha \left(1+\tfrac{\epsilon}{2}\right) x(L_{V\setminus C}) +\alpha \left(1+\tfrac{\epsilon}{2}\right) x(L_C) - \alpha \tfrac{\epsilon}{4} \cdot x(L_C) \notag\\
    &\ + \alpha \left(1+\tfrac{\epsilon}{4}\right) \cdot x(\delta_L(C)) \notag\\
    &\ + \max_{\substack{\mathcal{S}\!\text{ a $(\gamma',k)$-}\notag\\ \text{splitting of }\Iscr_{V\setminus C}}}\quad \sum_{\mathcal{J}\!\text{ split minor of }\mathcal{S}} f(\mathcal{J}) \notag\\ &\ + \max\{f(\Iscr') \colon \Iscr' \text{ is a $\lambda$-contraction minor for $\Iscr_C$}\} \notag\\[4mm]
    =&\ \alpha \left(1+\tfrac{\epsilon}{2}\right) \cdot x(L) + \alpha \left(2+\tfrac{3\epsilon}{4}\right) \cdot x(\delta_L(C)) - \alpha \tfrac{\epsilon}{4} \cdot x(L_C) \\
    &\ + \max_{\substack{\mathcal{S}\!\text{ a $(\gamma',k)$-}\notag\\ \text{splitting of }\Iscr_{V\setminus C}}}\quad \sum_{\mathcal{J}\!\text{ split minor of }\mathcal{S}} f(\mathcal{J}) \notag\\ &\ + \max\{f(\Iscr') \colon \Iscr' \text{ is a $\lambda$-contraction minor for $\Iscr_C$}\},\notag
\end{align}
\endgroup
where we used $L_C \cup L_{V\setminus C} = L$ and $L_C \cap L_{V\setminus C} = \delta_L(C)$.
Because $|C\cap T| > \tfrac{k}{2}$, we have $x(L_C) \ge \frac{k}{4} = 16\cdot\frac{8+3\epsilon}{\epsilon^3}$.
Moreover, $x(\delta_L(C))\le \frac{16}{\epsilon}$ because $\delta_L(C)$ is not $x$-heavy.
This implies
\begin{equation*}
 \left(2+\tfrac{3\epsilon}{4}\right) \cdot x(\delta_L(C))\ \le\  \left(2+\tfrac{3\epsilon}{4}\right) \cdot  \tfrac{16}{\epsilon}\  \le\ \tfrac{\epsilon}{4} \cdot 16\cdot\tfrac{8+3\epsilon}{\epsilon^3}\ \le\ \tfrac{\epsilon}{4}\cdot x(L_C).
\end{equation*}
Together with \eqref{eq:computation_minor}, we obtain
\begin{equation}\label{eq:contraction_minor}
\begin{aligned}
    |F \cup F_C \cup F_{V\setminus C}|\ \le&\ \alpha \cdot \left(1+ \tfrac{\epsilon}{2}\right) \cdot x(L)\\[2mm]
    &\ + \max\{f(\Iscr') \colon \Iscr' \text{ is a $\lambda$-contraction minor for $\Iscr_C$}\} \\
    &\ + \max_{\substack{\mathcal{S}\!\text{ a $(\gamma',k)$-}\\ \text{splitting of }\Iscr_{V\setminus C}}}\quad \sum_{\mathcal{J}\!\text{ split minor of }\mathcal{S}} f(\mathcal{J}).
\end{aligned}
\end{equation}
Let $\Iscr'_C \in \text{argmax}\{f(\Iscr') \colon \Iscr' \text{ is a $\lambda$-contraction minor for $\Iscr_C$}\}$. 
It remains to prove the following claim.

\begin{claim}\label{claim:bounded_contraction}
\begin{equation*}
f(\Iscr'_C) + \max_{\substack{\mathcal{S}\!\text{ a $(\gamma',k)$-}\\ \text{splitting of }\Iscr_{V\setminus C}}}\quad \sum_{\mathcal{J}\!\text{ split minor of }\mathcal{S}} f(\mathcal{J}) \leq \max_{\substack{\mathcal{S}\!\text{ a $(\gamma,k)$-}\\ \text{splitting of }\Iscr}}\quad \sum_{\mathcal{J}\!\text{ split minor of }\mathcal{S}} f(\mathcal{J}) .
\end{equation*}
\end{claim}
\begin{claimproof}
Let $\Jscr_1,\dots,\Jscr_N$ be the split minors of a $(\gamma,k)$-splitting of $\Iscr_{V\setminus C}$ that maximize 
$$\max_{\substack{\mathcal{S}\!\text{ a $(\gamma',k)$-}\\ \text{splitting of }\Iscr_{V\setminus C}}}\quad \sum_{\mathcal{J}\!\text{ split minor of }\mathcal{S}} f(\mathcal{J}) .$$
We finish the proof by showing that $\Iscr'_C,\Jscr_1,\dots,\Jscr_N$ are the split minors of a  $(\gamma,k)$-splitting of $\Iscr$. By construction, all $\Iscr'_C,\Jscr_1,\dots,\Jscr_N$ are obtained after applying a sequence of splitting and contraction operations and deletions of links. Moreover, they are all $k$-wide. We only need to prove that the number of contractions is at most $\gamma = \tfrac{\epsilon}{4}\cdot |T|$.

Recall we choose the cut $C$ to split at such that $|C\cap T| \ge \tfrac{k}{2}$, implying $|T_{V\setminus C}| \le |T| - \tfrac{k}{2} + 1$.
Therefore, because $\Jscr_1,\dots,\Jscr_N$ are the split minors of a $(\gamma',k)$-splitting of $\Iscr_{V\setminus C}$, the number of contractions applied to obtain $\Jscr_1,\dots,\Jscr_N$ from $\Iscr_{V\setminus C}$ (by splitting, contraction, and deletion of links) is at most $\gamma' = \tfrac{\epsilon}{4} \cdot |T_{V\setminus C}| \le \tfrac{\epsilon}{4} \cdot (|T| - \tfrac{k}{2} + 1)$.
The total number of contractions to obtain $\Iscr'_C,\Jscr_1,\dots,\Jscr_N$ is thus at most 
\[
\lambda + \tfrac{\epsilon}{4} \cdot (|T| - \tfrac{k}{2} + 1) \le \tfrac{\epsilon}{4} \cdot |T| = \gamma
\]
because
\[
\lambda +\frac{\epsilon}{4}=  \frac{1+\epsilon'}{\epsilon'}\cdot \frac{16}{\epsilon}+\frac{\epsilon}{4} =\frac{1+\sfrac{\epsilon}{4}}{\sfrac{\epsilon}{4}}\cdot \frac{16}{\epsilon}+\frac{\epsilon}{4} \leq \frac{\epsilon}{8}\cdot \frac{64(8+3\epsilon)}{\epsilon^3}= \frac{\epsilon}{4}\cdot\frac{k}{2} .
\]
\end{claimproof}
\end{proof}

In order to obtain an algorithm that, given a point $x\in [0,1]^L$, returns either
 a CacAP solution as in \ref{item:round} or a vector $w\in \mathbb{R}^L$ as in \ref{item:seperate}, we proceed as follows.
First, we observe that if $x(L) \le \sfrac{|T|}{2}$ we can return a vector $w$ as in \ref{item:seperate} because $x'(L) \ge \sfrac{|T|}{2}$ for every vector $x'\in P_{\mathrm{CacAP}}(\mathcal{I})  $.
Hence, we will assume in the following that this is not the case.

Given $x\in[0,1]^L$ and $\Wscr = \{C\in \mathcal{C} : C\text{ is $x$-heavy}\}$, we apply \cite[Theorem~23]{cecchetto_2021_bridging} to obtain a cheap heavy cut covering, i.e., a set $L_H\subseteq L$ of links covering all heavy cuts with $|L_H|\leq \frac{\epsilon}{2}\cdot x(L)$.
Then we consider the residual instance $\Iscr_H$ of $\Iscr$ with respect to $L_H$ (see \cref{def:residual_instance}) and apply \cref{lem:round_and_cut_without_heavy}.
Because $\Iscr_H$ arises from $\Iscr$ by the contraction of $|L_H|$ many links, every $(\gamma, k)$-splitting of $\Iscr_H$ is a $(\gamma+|L_H|, k)$-splitting of $\Iscr$.
Hence, the application of \cref{lem:round_and_cut_without_heavy} yields either a separating hyperplane or a CacAP solution $F$ for $\Iscr_H$ with 

\begin{equation*}
    |F|\le \alpha \cdot \left(1+\frac{\epsilon}{2}\right) \cdot x(L) +\  \max_{\substack{\mathcal{S}\!\text{ an $(\epsilon \cdot x(L) ,k)$-}\\ \text{splitting of }\Iscr}}\quad \sum_{\mathcal{J}\!\text{ split minor of }\mathcal{S}} f(\mathcal{J}) ,
\end{equation*}
where we used $\gamma + |L_H| \le \tfrac{\epsilon}{4} \cdot |T| + \frac{\epsilon}{2}\cdot x(L) \le \epsilon \cdot x(L)$. 
Then \cite[Corollary~20]{cecchetto_2021_bridging} implies that $F\cup L_H$ is  a solution for the instance $\Iscr$ with the desired guarantee~\ref{item:round}.

\section{Combining our matching-based approach with the stack analysis (Proof of \texorpdfstring{\cref{thm:main}}{Theorem 1})}\label{sec:stack_analysis}

To obtain approximation factors below $\tfrac{4}{3}$, we use the stack analysis approach from \cite{cecchetto_2021_bridging}, which strengthens the procedure guaranteed by \cref{lem:fpt_in_nb_leaves}. 
In this section we explain how we can adapt this stack analysis approach from \cite{cecchetto_2021_bridging} for our purposes and we prove that it leads to the claimed approximation ratio of $\apxfac$.

Let us first recall the definition of \emph{shadows} and  \emph{minimality} of links from \cite{cecchetto_2021_bridging}.

\begin{definition}
Let $(G,L)$ be a CacAP instance and let $\{u,v\}$ be a link.
Then $\{\bar u, \bar v\}\in \begin{psmallmatrix} V\\ 2 \end{psmallmatrix}$ is a \emph{shadow} of the link $\{u,v\}$ if $\bar u$ and $\bar v$ are vertices that lie on every $u$-$v$ path in $G$.
A shadow $\bar \ell$ of link $\ell$ is a \emph{strict shadow} of $\ell$ if $\bar \ell \ne \ell$.
\end{definition}

Following \cite{cecchetto_2021_bridging}, we say that a link $\ell_1\in L$ is \emph{minimal with respect to $\ell_2\in L$} if for any strict shadow $\bar{\ell_1}\in \begin{psmallmatrix} V\\ 2 \end{psmallmatrix}$ of $\ell_1$, the $2$-cuts covered by $\{\bar{\ell_1},\ell_2\}$ are a strict subset of those covered by $\{\ell_1,\ell_2\}$ and the 2-cuts covered by $\{\ell_2\}$ are a strict subset of those covered by $\{\ell_1,\ell_2\}$; or formally, for any strict shadow $\bar{\ell_1}$ of $\ell_1$,
\begin{equation*}
\begin{aligned}
\{ C\in \mathcal{C} : \{\bar\ell_1,\ell_2\} \cap \delta(C) \ne \emptyset\} \subsetneq&\ \{ C  \in \mathcal{C} : \{\ell_1, \ell_2\} \cap \delta(C) \ne \emptyset\},\text{ and} \\
\{ C\in \mathcal{C} : \{\ell_2\} \cap \delta(C) \ne \emptyset\} \subsetneq&\ \{ C  \in \mathcal{C} : \{\ell_1, \ell_2\} \cap \delta(C) \ne \emptyset\}
.
\end{aligned}
\end{equation*}
We remark that the second of the two conditions above is implied by the first one whenever $\ell_1$ admits a strict shadow.
However, for the case where $\ell_1$ does not admit a strict shadow, which corresponds to both endpoints of $\ell_1$ lying in the same cycle of the cactus, the second condition is necessary.

We now introduce the notion of \emph{weakly $L_{\rm{cross}}$-minimality}.
A similar notion, called $L_{\rm{cross}}$-minimality, has been introduced in \cite{cecchetto_2021_bridging}.
The reason why we work with the notion of weak $L_{\rm{cross}}$-minimality is that we will need this later to combine the stack analysis approach with our matching-based approach from \cref{sec:matching_approach}.

\begin{definition}
A set $F\subseteq L$ is called weakly 
$L_{\mathrm{cross}}$-minimal if for every two distinct links $\ell,\ell'\in F_{\mathrm{cross}}$ the link $\ell$ is minimal with respect to $\ell'$.
\end{definition}

The property of \emph{$L_{\rm{cross}}$-minimality} introduced in \cite{cecchetto_2021_bridging} requires in addition that every cross-link $\ell_1\in F_{\mathrm{cross}}$ is minimal with respect to every in-link $\ell_2 \in F_{\mathrm{in}}$.
However, for leaf-to-leaf instances, the notion of weak $L_{\rm{cross}}$-minimality will be sufficient for the stack analysis approach.

The stack analysis approach first solves a linear program and then rounds the LP solution to a CacAP solution.
We will next introduce the polytope over which this linear program optimizes.

For a $k$-wide CacAP instance $(G,L)$ with root $r$ let $G_1,\dots, G_p$ denote the principal subcacti.
Moreover, for $i\in\{1,\dots,p\}$, we define $L_i \subseteq L$ to be the set of links that have at least one endpoint in $G_i$ that is distinct from the root $r$.
We define
\[
P^{\min}(G_i,L_i) \coloneqq \mathrm{conv}\left(\{ \chi^F : F \subseteq L_i\text{ is a weakly  $L_{\mathrm{cross}}$-minimal feasible solution for }G_i\}\right)
\]
to be the convex hull of incidence vectors of weakly $L_{\mathrm{cross}}$-minimal feasible CacAP solutions for $G_i$, where a feasible  CacAP solutions for $G_i$ is a set of links that covers every $C\in \mathcal{C}$ that is a subset of the vertices of $G_i$. 
The same proof as for \cite[Lemma~42]{cecchetto_2021_bridging} shows that we can optimize over $P^{\min}(G_i)$ in polynomial time when $k$ is constant.
Indeed, we can first observe that any weakly $L_{\mathrm{cross}}$-minimal solution for $(G_i,L_i)$ contains at most $k$ cross-links and we can ``guess'' those in polynomial time by enumerating all possible choices.
Then we can complete this set of cross-links in a cheapest possible way, using that instances of CacAP with a constant number of terminals are efficiently solvable \cite{basavaraju_2014_parameterized}.

This implies that we can also optimize efficiently over the polytope
\begin{equation*}
P^{\mathrm{min}}_{\mathrm{bundle}}(G,L) \coloneqq \left\{ x\in [0,1]^L : x|_{L_i} \in P^{\min}(G_i,L_i) \text{ for all }i\in \{1,\dots,p\} \right\}
\end{equation*}
because the fact that we can efficiently optimize over $P^{\min}(G_i,L_i)$ implies that we can separate over $P^{\min}(G_i,L_i)$ and thus over $P^{\mathrm{min}}_{\mathrm{bundle}}(G,L) $.
For a more detailed description of how we can optimize over $P^{\mathrm{min}}_{\mathrm{bundle}}(G,L)$ we refer to \cite{cecchetto_2021_bridging} where an analogous reasoning has been used.

\cite{cecchetto_2021_bridging} implies that we can round vectors in $P^{\mathrm{min}}_{\mathrm{bundle}}(G,L)$ for $k$-wide instances of Leaf-to-Leaf+ CacAP with a certain guarantee on the cost of the CacAP.
Formally, \cite{cecchetto_2021_bridging} considered general $k$-wide instances of CacAP and defined $P^{\mathrm{min}}_{\mathrm{bundle}}(G,L)$ in a slightly different way, requiring that the vectors $x|_{L_i}$ for $i\in \{1,\dots,p\}$ are convex combinations of incidence vectors of $L_{\mathrm{cross}}$-minimal solutions (instead of weakly $L_{\mathrm{cross}}$-minimal solutions).
However, this stronger property is used only in a single place of the proof, namely the proof of Lemma~52 in \cite{cecchetto_2021_bridging}, which is a trivial statement for leaf-to-leaf+ instances.\footnote{Lemma~52 from \cite{cecchetto_2021_bridging} provides a lower bound on $x(L_{\mathrm{up}})$ where $L_{\mathrm{up}} \subseteq L$. 
This lower bound is $0$ for every leaf-to-leaf+ instance and thus the statement of Lemma~52 from \cite{cecchetto_2021_bridging} trivially holds in our setting.}
Therefore, the guarantees from \cite{cecchetto_2021_bridging} still apply in our setting despite the slightly different definition of $P^{\mathrm{min}}_{\mathrm{bundle}}(G,L)$.

The precise guarantee on the solution that we obtain by rounding a point $x\in P^{\mathrm{min}}_{\mathrm{bundle}}(G,L)$ is given by the optimal value of some optimization problem, which we will state later in this section (in the proof of \cref{thm:main}).

Note that in general, $P^{\mathrm{min}}_{\mathrm{bundle}}(G,L)$ is not a relaxation of the CacAP problem and it might even be the case that the polyope $P^{\mathrm{min}}_{\mathrm{bundle}}(G,L)$ is empty even though the instance $(G,L)$ is feasible.
The reason is that in general, not every instance $(G,L)$ has a weakly $L_{\rm{cross}}$-minimal solution.
However, we will show that every \emph{root-shadow complete} instance of Leaf-to-Leaf+ CacAP has a weakly $L_{\rm{cross}}$-minimal optimum solution (see \cref{lem:L-cross-minimality} below).

\begin{definition}
An instance $(G,L)$ of Leaf-to-Leaf+ CacAP is \emph{root-shadow complete} if for every cross link $\{u,v\}\in L_{\rm{cross}}$, both $\{u,r\}$ and $\{r,v\}$ are contained in $L$.
Then $\{u,r\}$ and $\{r,v\}$ are the \emph{root shadows} of $\{u,v\}$.
\end{definition}

We may always assume that the leaf-to-leaf+ instance that we are given is root-shadow complete because, given an arbitrary leaf-to-leaf+ instance $(G,L)$, we can consider its \emph{root-shadow completion}
\[
 \bigl(G,L \cup \bigl\{\{u,r\}, \{v,r\} : \{u,v\}\in L_{\rm{cross}}\bigr\}\bigr)
\]
that we obtain by adding all root-shadows of cross-links.
Given any solution to the root-shadow completion we can always turn it into a solution of the original instance with the same number of link by replacing every root shadow by the original link.
We remark that the root-shadow completion of a leaf-to-leaf+ instance is again a leaf-to-leaf+ instance.
However, the root-shadow completion of a pure leaf-to-leaf instance is not a leaf-to-leaf instance and this is the reason why we work with leaf-to-leaf+ instances in this paper.

\begin{lemma}\label{lem:L-cross-minimality}
Every root-shadow complete instance of Leaf-to-Leaf+ CacAP has a weakly $L_{\rm{cross}}$-minimal optimum solution.
\end{lemma}
\begin{proof}
Let $\OPT$ be a an optimum solution of a Leaf-to-Leaf+ CacAP instance with a minimum number of cross-links among all optimum solutions.
We claim that $\OPT$ is weakly $L_{\rm{cross}}$-minimal.
Suppose this is not the case.
Then there exist cross-links $\ell_1,\ell_2 \in \OPT_{\mathrm{cross}}$ such that $\ell_1$ is not minimal with respect to $\ell_2$.
Let $\ell_1=\{s,t\}$ and note that any strict shadow of $\ell_1$, i.e., any shadow $\bar{\ell_1}\neq \ell_1$ of $\ell_1$, covers at most one of the 2-cuts $\{s\},\{t\} \in \Cscr$.
Thus, one of the endpoints of $\ell_1$, say $t$, must be also an endpoint of $\ell_2$, because otherwise for every strict shadow $\bar \ell_1$ of $\ell_1$ at least one of the 2-cuts $\{s\},\{t\} \in \Cscr$ would be uncovered by $\{\bar \ell_1,\ell_2\}$, even though it was covered by $\{\ell_1,\ell_2\}$.
Because $\ell_1=\{s,t\}$ and $\ell_2$ are both cross-links and have the common endpoint $t$, the set of 2-cuts in $\Cscr$ covered by $\{s,r\}$ and $\ell_2$ is the same as the set of 2-cuts covered by $\ell_1$ and $\ell_2$.
Hence, we can replace the link $\ell_1$ in $\OPT$ by its root-shadow $\{s,r\}$ and maintain an optimum solution.
However, this replacement decreased the number of cross-links, contradicting our choice of $\OPT$.
\end{proof}

\cref{lem:L-cross-minimality} implies that for any root-shadow complete instance $(G,L)$ of CacAP, we have $\min\{ x(L) : x\in P^{\mathrm{min}}_{\mathrm{bundle}}(G,L) \} \le |\OPT|$ because $\chi^{\OPT}\in P^{\mathrm{min}}_{\mathrm{bundle}}(G,L)$ for any weakly $L_{\rm{cross}}$-minimal optimum solution $\OPT$.

We now show how we combine our matching-based approach from \cref{sec:matching_approach} with the stack analysis approach from \cite{cecchetto_2021_bridging} to prove \cref{thm:main}.

\mainthm*
\begin{proof}
By \cref{thm:reduction}, it suffices to show that there is a $\rho$-approximation algorithm for $O(1)$-wide instances of Leaf-to-Leaf CacAP for some $\rho < \apxfac$.

Let us now describe such a $\rho$-approximation algorithm for $O(1)$-wide instances of Leaf-to-Leaf CacAP.
For a $k$-wide instance of Leaf-to-Leaf CacAP, we consider its root-shadow completion $(G,L)$, which is a Leaf-to-Leaf+ CacAP instance.
Let $\OPT$ be a weakly $L_{\rm{cross}}$-minimal optimum solution of the instance $(G,L)$, which exists by \cref{lem:L-cross-minimality}.

We compute a matching $M\subseteq L$ on the leaves of $G$ without bad links  that minimizes
$|M| + \frac{1}{2}|M_{\mathrm{in}}| + (|T|-2|M|)$.
By \cref{lem:matching_bound}, we have
\begin{equation*}
|M| + \frac{1}{2}|M_{\mathrm{in}}| + (|T|-2|M|)  \leq |\OPT|+\frac{1}{2}|\OPT_{\rm{in}}|.
\end{equation*}
Because $\OPT$ is weakly $L_{\rm{cross}}$-minimal, this implies that the incidence vector $\chi^{\OPT}$ of $\OPT$ is a feasible solution to the following linear program:
\begin{equation} \label{lp:combiningprocedures}
\renewcommand\arraystretch{1.5}
\begin{array}{r>{\displaystyle}rc>{\displaystyle}ll}
\min  & x(L)  & & & \\
\text{s.t.}& |M| + \frac{1}{2}|M_{\mathrm{in}}| + (|T|-2|M|)  & \leq & x(L)+\frac{1}{2}x(L_{\mathrm{in}}) & \\
& x & \in & P_{\mathrm{bundle}}^{\mathrm{min}}(G,L) .
\end{array}
\end{equation}
We compute an optimum solution $x$ to the LP~\eqref{lp:combiningprocedures}.
Because $\chi^{\OPT}$ is a feasible solution to~\eqref{lp:combiningprocedures}, we have $x(L) \le |\OPT|$.

By \cref{lem:matching_bound}, we can compute a feasible CacAP solution $F_1$ with 
\begin{equation}\label{eq:guarantee_matching_rounding}
|F_1| \ \le\ |M| + \tfrac{1}{2} |M_{\rm{in}}| + (|T| -2|M|)\ \leq\ x(L)+\frac{1}{2}x(L_{\mathrm{in}}).
\end{equation}
Finally, we apply the rounding procedure from the stack analysis approach from \cite{cecchetto_2021_bridging} to obtain a solution $F_2$ and return the cheaper of the two solutions $F_1$ and $F_2$.

In order to state the guarantee on $|F_2|$ that we obtain from \cite{cecchetto_2021_bridging}, we need the following definitions.
We define the function $g : [0,1] \to \mathbb{R}_{\ge 0}$ as 
\begin{equation*}
    g(\lambda) \coloneqq \lambda \cdot \big(1- e^{-\lambda}\big)
\end{equation*}
and the function $\mathrm{gain}: [0,1]^2 \to \mathbb{R}$ by
\begin{equation*}
 \mathrm{gain}(\lambda, \eta) \coloneqq \begin{cases}
   \lambda \left( e^{- \eta} -1 + \eta \right)\cdot e^{- \lambda + \eta} &\text{ if }\eta > \frac{1}{2}\lambda \\ 
   \lambda \left( e^{- \eta} -1 + \eta \right)\cdot\left(1- \lambda + \eta\right) & \text{ otherwise} .
 \end{cases}
\end{equation*}
Moreover, for a leaf $t\in T$, we define $\lambda_t^0 \coloneqq x(\delta(t) \cap L_{\rm{cross}})$.
By \cite{cecchetto_2021_bridging},\footnote{More precisely, the statement we use here follows essentially from Lemma~58 in~\cite{cecchetto_2021_bridging}. Note that Condition~\eqref{eq:b_condition_checked_numerically} is for leaf-to-leaf+ instances equivalent to equation~(27) in \cite{cecchetto_2021_bridging}.
The condition  $b\in [\frac{5}{12},\frac{1}{2}]$ is equivalent to the condition $z_1\le z_2 \le 0 $ below equation~(18) in \cite{cecchetto_2021_bridging}.
We remark that in a leaf-to-leaf+ instance, the values $\lambda_t^1$, $\mu^0_t$, and $\mu^1_t$ appearing in \cite{cecchetto_2021_bridging} are all equal to zero by definition of these values and thus we eliminated their occurrences here.
In particular, our definitions of the functions $g$ and $\mathrm{gain}$ then coincide with those from \cite{cecchetto_2021_bridging}.
Finally, we again highlight that we work with a slightly different definition of $P_{\mathrm{bundle}}^{\mathrm{min}}(G,L)$ than \cite{cecchetto_2021_bridging} but the only place where the stronger formulation in \cite{cecchetto_2021_bridging} is needed is in the proof of Lemma~52 in \cite{cecchetto_2021_bridging}, which is trivial for leaf-to-leaf+ instances.} 
we can round any vector $x\in P_{\mathrm{bundle}}^{\mathrm{min}}(G,L)$ to a CacAP solution $F_2$ such that 
\begin{equation}\label{eq:stack_algo}
    |F_2| \leq x(L_{\mathrm{in}}) + 2 \cdot x(L_{\mathrm{cross}}) - b \cdot \sum_{t\in T} g(\lambda_t^0)
\end{equation}
if $b\in [\frac{5}{12},\frac{1}{2}]$ is such that the following expression is non-negative for all $v,w \in T$
\begin{equation}\label{eq:b_condition_checked_numerically}
\begin{aligned}
&\ \frac{b}{\lambda_w^0 - \eta_w^{c_v}} 
         \cdot  \mathrm{gain}(\lambda_w^0, \eta^{c_v}_w)  \\
&\ - s_{vw} \cdot  \left(b-\tfrac{1}{3}\right) - (\eta_w^{c_v} - s_{vw}) \cdot \left(2\left(b-\tfrac{2}{5}\right)-\tfrac{1}{30}\right) \\
&+ \max\{0,\ x(S_v) - \eta_w^{c_v}\} \cdot \left( \tfrac{1}{2}-b \right) + \max\{0,\ 1 - x(S_v) -\eta_w^{c_v} +s_{vw}\} \cdot \left(1-b\right),
\end{aligned}
\end{equation}
for all $s_{vw}$ ($v,w\in T$) and $\eta^{c_v}_w$ ($v,w\in T$) that fulfill
\begin{equation*}
    \begin{aligned}
        & 0 \leq s_{vw} \leq \eta_w^{c_v} \leq \lambda_w^0 \leq 1 \\
        & 0 \leq s_{vw} \leq \lambda_v^0 \leq 1  .
    \end{aligned}
\end{equation*}
Using a computer program, we can verify that Condition~\eqref{eq:b_condition_checked_numerically} is fulfilled for $b\coloneqq 0.452$.

Let $\alpha \coloneqq \frac{x(L{\mathrm{cross}})}{x(L)}$.
Then, by~\eqref{eq:stack_algo}, the CacAP solution $F_2$ obtained from the stack analysis approach fulfills
\begin{equation}\label{eq:result_stack_algo}
\begin{aligned}
|F_2|\ \le&\ x(L_{\mathrm{in}}) + 2 \cdot x(L_{\mathrm{cross}}) - b \cdot \sum_{v\in T} g(\lambda^0_v) \\
=&\ \left(1 + \alpha - \frac{b}{x(L)} \sum_{v\in T} g(\lambda^0_v)  \right) \cdot x(L) \\
=&\ \left(1 + \alpha - \frac{b \cdot 2\alpha}{\sum_{v\in T} \lambda^0_v } \sum_{v\in T} g(\lambda^0_v)  \right) \cdot x(L) . 
\end{aligned}
\end{equation}

The solution we return has $\min\{ |F_1|, |F_2|\}$ many links and thus by combining \eqref{eq:guarantee_matching_rounding}, $x(L_{\mathrm{in}}) = (1 - \alpha) \cdot x(L)$, and \eqref{eq:result_stack_algo}, we obtain an approximation ratio of
\begin{equation*}
\max \left\{\ \min\left\{\ \frac{3}{2}-\frac{1}{2}\alpha, 1 + \alpha - \frac{b \cdot 2 \alpha}{\sum_{v\in T} \lambda^0_v } \sum_{v\in T} g(\lambda^0_v)\ \right\} :
\begin{array}{l}
     \alpha \in [0,1], \ \lambda^0_v \in [0,1] \ \forall \ v\in T, \\
      \alpha\cdot |T|\leq \sum_{v\in T} \lambda^0_v \leq |T| 
\end{array}\right\} .
\end{equation*}
Because the function $g$ is convex in $\lambda$, replacing $\lambda^0_v$ by its average value over all of the leaves $v\in T$,
does not increase the value of $\sum_{v\in T} g(\lambda^0_v)$.
Moreover, this replacement does not change $\sum_{v\in T} \lambda^0_v$.
Thus, we can simplify the optimization problem and conclude that we obtain an approximation ratio of 
\begin{equation*}
\rho \coloneqq \max \Big\{\ \min\Big\{\ \frac{3}{2}-\frac{1}{2}\alpha, 1 + \alpha - \frac{b \cdot 2 \alpha}{\lambda^0}\cdot  g(\lambda^0)\ \Big\}\  : \ 0\leq \alpha \leq \lambda^0 \leq 1 \Big\}.
\end{equation*}
As $g(\lambda) = \lambda \big( 1-e^{-\lambda}\big)$ and $b=0.452$, this yields
\begin{equation}\label{eq:final_optimization_problem}
\rho = \max \Big\{\ \min\Big\{\ \frac{3}{2}-\frac{1}{2}\alpha, 1 + \alpha - 0.904\alpha \cdot \big( 1-e^{-\lambda^0}\big) \ \Big\}\  : \ 0\leq \alpha \leq \lambda^0 \leq 1 \Big\}.
\end{equation}
The optimum value of this optimization problem is attained for $\alpha=\lambda^0$ being the unique solution to $ 6\alpha + 9\alpha e^{-\alpha} =5$.
This yields $\alpha = \lambda^0= 0.4202$ and $\rho < \apxfac$.
\end{proof}
 
\end{document}